\documentclass[12pt]{article}
\RequirePackage{amsthm,amsmath, amssymb}
\usepackage{graphicx,psfrag,epsf}
\usepackage{enumerate}
\usepackage[numbers]{natbib}
\usepackage{algorithm}
\usepackage{algpseudocode}
\usepackage{amsthm,amsmath,amsfonts}
\usepackage{graphicx}
\usepackage{amsmath,amssymb,amsthm,amsopn,bm}
\usepackage{soul}
\usepackage[mathscr]{euscript}
\usepackage{float}
\usepackage{ulem}
\makeatletter
\renewcommand{\ALG@name}{Pseudo-algorithm}
\makeatother

\numberwithin{equation}{section}
\theoremstyle{plain}

\newtheorem{lemma}{Lemma}
\newtheorem{theorem}{Theorem}
\newtheorem{corollary}{Corollary}
\theoremstyle{remark}
\newtheorem{remark}{Remark}

\newcommand{\blind}{0}

\begin{document}

\def\spacingset#1{\renewcommand{\baselinestretch}%
{#1}\small\normalsize} \spacingset{1}


\if0\blind
{
  \title{\bf Modal clustering asymptotics with applications to bandwidth selection}
  \author{Alessandro Casa\hspace{.2cm}\\
    Department of Statistical Sciences, University of Padova\vspace{.3cm}\\
    Jos\'e E. Chac\'on \\
    Department of Mathematics, University of Extremadura \vspace{.1cm}\\
    and \vspace{.1cm}\\
    Giovanna Menardi \\
    Department of Statistical Sciences, University of Padova}
  \maketitle
} \fi

\bigskip
\begin{abstract}
Density-based clustering relies on the idea of linking groups to some specific features of the probability distribution underlying the data. The reference to a true, yet unknown, population structure allows to frame the clustering problem in a standard inferential setting, where the concept of ideal population clustering is defined as the partition induced by the true density function. The nonparametric formulation of this approach, known as modal clustering, draws a correspondence between the groups and the domains of attraction of the density modes. Operationally, a nonparametric density estimate is required and a proper selection of the amount of smoothing, governing the shape of the density and hence possibly the modal structure, is crucial to identify the final partition. In this work, we address the issue of density estimation for modal clustering from an asymptotic perspective. A natural and easy to interpret metric to measure the distance between density-based partitions is discussed, its asymptotic approximation explored, and employed to study the problem of bandwidth selection for nonparametric modal clustering.
\end{abstract}

\noindent%

\spacingset{1.45}
\section{Introduction}

Clustering is commonly referred to as the task of finding groups in a set of data points (see \cite{KR05}, \cite{Eal11} or \cite{Hal16}). While intuitively clear, this task is, in fact, far from being accurately defined. The density-based approach attempts to circumscribe this issue by framing the problem into a statistically rigorous setting where the observed data are assumed to be realizations of a random variable, and the clusters are defined with respect to some characteristic of its underlying probability distribution.

In this sense, a clustering procedure should not be limited to simply produce a partition of the observed data; instead, it must allow to obtain a \textit{whole-space clustering}, that is a partition of the whole sample space \citep{Bal06, chacon15}. 
In any case, each methodology is characterized by the way in which the clusters are defined in terms of the true distribution, leading to the concept of \textit{ideal population clustering}. By serving as a reference ``ground truth'' to aim at, this concept introduces a benchmark to evaluate the performance of data-based partitions. 

The ideal population goal in density-based clustering can be defined in terms of two different paradigms: the \textit{model-based} approach, where each cluster is associated to a parametric mixture component, and the \textit{modal} one (see respectively \cite{mcnicholas16} and \cite{menardi16} for some recent reviews). This paper focuses on the latter formulation, whose name stems from the notion of clusters as the ``domains of attraction" of the modes of the true density underlying the data \cite{stuetzle03}.


Therefore, in practice density estimation assumes a key role in order to approximate the ideal population goal of modal clustering. While the modal formulation does not preclude using a parametric density estimate as a first step to perform a data-based modal clustering \cite{chacon19, scrucca16}, a long-standing practice resorts to nonparametric estimators. Precisely, in this paper the focus lies on those estimators based on kernel smoothing (see e.g. \cite{chaconduongbook} and \cite{wandjones}).

Under- or over-smoothed estimates may lead to deceiving indications about the modal structure of the underlying density function, and this problem is usually quantified through some measure of the discrepancy between the estimate and the target density. In contrast, the aim of this work is to consider nonparametric density estimation as a tool for the final purpose of modal clustering, focusing on an appropriate metric comparing the partitions induced by the true and the estimated distribution.

Our main result provides an asymptotic approximation for the considered metric, which allows to introduce new automatic bandwidth selection procedures specifically designed for nonparametric modal clustering. The accuracy of this approximation and the performance of the new methods in practice, with respect to the proposed error criterion, is extensively studied via simulations, and compared with some plausible competitors.

The rest of the paper is structured as follows. Section \ref{sec:background} formally introduces the modal approach to cluster analysis with reference also to algorithmic details. In Section \ref{sec:mainsec} the distance criterion to target density estimation for modal clustering is presented, along with the main asymptotic result and its consequences. Section \ref{sec:numresults} contains the setup and results of the numerical experiments. A generalization to the multidimensional setting is discussed in Section \ref{sec:multidim}. Finally, some concluding remarks are stated in Section \ref{sec:conclusion}.

\section{Background}\label{sec:background}{}

The connection between groups and density features, established by the modal approach to cluster analysis, allows to characterize the concept of ideal population clustering. Informally, a population cluster can be defined as the {domain of attraction} of a mode of the density \cite{stuetzle03}.
An attempt to formalize this concept has been done in \cite{chacon15} with the aid of Morse Theory, a branch of differential topology focusing on the large scale structure of an object via the analysis of the critical points of a function (see e.g. \cite{matsumoto02} for an introduction).

Let us consider a continuous $d$-variate random variable $X$, with probability density function $f\colon\mathbb{R}^d \rightarrow \mathbb{R}$. Assume that $f$ is a Morse function, i.e. a smooth enough function having nondegenerate critical points, and denote by $M_1,\dots, M_r$ the modes of $f$ (i.e. its local maxima).
For a given initial value $x \in \mathbb{R}^d$, an \textit{integral curve} of the negative density gradient $-\nabla f$ is defined as the path $\nu_x : \mathbb{R} \rightarrow \mathbb{R}^d$ such that
\begin{eqnarray*}
\nu'_x(t) = -\nabla f(\nu_x(t)), \hspace{0.5cm} \nu_x(0)=x.
\end{eqnarray*}
The set of points whose integral curve starts at a critical point $x_0$ (as $t\to -\infty$) goes under the name of \textit{unstable manifold} of $x_0$ and is defined as
\begin{eqnarray*}
 W^u_-(x_0) = \{ x \in \mathbb{R}^d : \lim_{t\rightarrow - \infty} \nu_x(t)=x_0  \}.
\end{eqnarray*}
It has been showed \cite{thom49} that the class of the unstable manifolds of every critical point of a Morse function yields a partition of the whole space.
With these notions at hand, the ideal population clustering $\mathcal{C}=\{ \mathcal{C}_1,\dots,\mathcal{C}_r \}$ associated to a density function $f$ is then defined as the set of the unstable manifolds $\{W_{-}^u(M_1),\dots W_{-}^u(M_r) \}$ of the modes of $f$. By borrowing concepts from terrain analysis, the underlying intuition is that, if $f$ is figured as a mountainous landscape where the modes are the peaks, a modal cluster is the region that would be flooded by a fountain emanating from a peak. When $d=1$, clusters are then unequivocally defined by the locations of the minima points of $f$, which represent the cluster boundaries.

Equivalently, if the integral curves associated to the positive density gradient are considered, then a modal cluster is defined as the set of points whose integral curves converge (as $t\to+\infty$) at the same mode. The concept of modal clusters as the domains of attraction of the density modes stems naturally from this definition. Operationally, a numerical algorithm is needed to find the eventual destination of an initial point, and most of the contributions in this direction take their steps from the mean-shift algorithm \cite{fukunaga75}, essentially a variant of the gradient ascent algorithm. The algorithm transforms an initial point $x^{(0)}$ recursively, and identifies a sequence $(x^{(0)},x^{(1)},x^{(2)},\dots)$ according to an updating mechanism defined as
\begin{eqnarray*}
x^{(l+1)}=x^{(l)}+A \frac{\nabla f(x^{(l)})}{f(x^{(l)})} \; ,
\end{eqnarray*}
where $A$ is a $d\times d$ positive definite matrix chosen to guarantee the convergence to a local maximum of $f$. A partition of the data is therefore obtained by simply grouping together the observations climbing to the same density mode, via mean-shift updates.

From a practical point of view the density $f$ is unknown, therefore an estimate is needed. When working in a nonparametric framework a common choice is given by the kernel density estimator. In the following we focus on the univariate case for ease of exposition and mathematical tractability while the multivariate extension will be addressed in Section \ref{sec:multidim} below. Let $X_1,\dots,X_n$ be a sample of i.i.d. realizations of $X$. Then, the kernel density estimator is defined by
\begin{eqnarray*}
\hat{f}_h(x) = \frac{1}{nh}\sum_{i=1}^n K\left( \frac{x-X_i}{h} \right) \; ,
\end{eqnarray*}
where $K$ is the kernel, usually a smooth, non-negative and symmetric function integrating to one, and $h$ is the bandwidth, which controls the smoothness of the density estimate.

While the choice of the function $K$ is known not to have a strong impact in the performance of the estimate \citep[][Section 3.3.2]{Silverman86}, choosing $h$ properly turns out to be crucial. A small value of $h$ leads to an undersmoothed density estimate, with the possible appearence of spurious modes, while a too large value results in an oversmoothed density estimate, possibly hiding relevant features. 

In order to select the smoothing parameter some measure of the distance between the estimated and the true density is needed. A common choice is the \textit{Integrated Squared Error}, defined as
\begin{eqnarray*}
{\rm ISE}(h) = \int_\mathbb{R} \{ \hat{f}_h(x)- f(x) \}^2 dx.
\end{eqnarray*}
Depending on the observed data, the ISE is itself subject to a random variability that could hinder the problem of bandwidth selection (see \cite{HM91}). Hence, its expected value
\begin{equation}\label{eq:MISE}
{\rm MISE}(h)=\mathbb{E} \left[{\rm ISE}(h) \right]
\end{equation}
is alternatively considered as a non-stochastic error distance. The optimal bandwidth $h_{\rm MISE}$ is then defined as $h_{\rm MISE} = \mathop{\rm argmin}_{h>0} {\rm MISE}(h)$.

Since minimization of the MISE does not lead to closed form solutions for the optimal bandwidth, its asymptotic counterpart -- the AMISE -- is often considered.
Both the MISE and the AMISE depend on the true, unknown density function; for this reason several different approaches to estimate them have been proposed. Examples are the ones based on \textit{least squares cross validation}, \textit{biased cross validation} or \textit{plug-in bandwidth selectors}. A comprehensive review of these methods is beyond the scope of this work and, for a complete exposition, readers can refer to \cite{wandjones} or to the more recent book by \cite{chaconduongbook}.

\section{Density estimation for modal clustering}\label{sec:mainsec}
\subsection{Asymptotic bandwidth selection for modal clustering}\label{sec:bwselect}

Bandwidth selectors based on the ISE or akin distances pursue the aim of obtaining an appropriate estimate of the density. However, the goal of modal clustering is markedly different from that of density estimation (see e.g. \cite{cuevasetal}). In fact, two densities that are close with respect to the ISE may result in quite different clusterings while, on the other hand, densities far away from an ISE point of view could lead to the same partition of the space. A graphical illustration of this idea is provided by Figure \ref{fig:fig1}. The inappropriateness of the ISE, or related distances, depends on its focus on the global characteristics of the density, while modal clustering strongly builds on specific and local features, more closely related to the density gradient or the high-density regions (see also \cite{cgw2017}). Therefore, the choice of the amount of smoothing should be tailored specifically for clustering purposes.

\begin{figure}[tb]
\centering
\includegraphics[height=5cm,width=12cm]{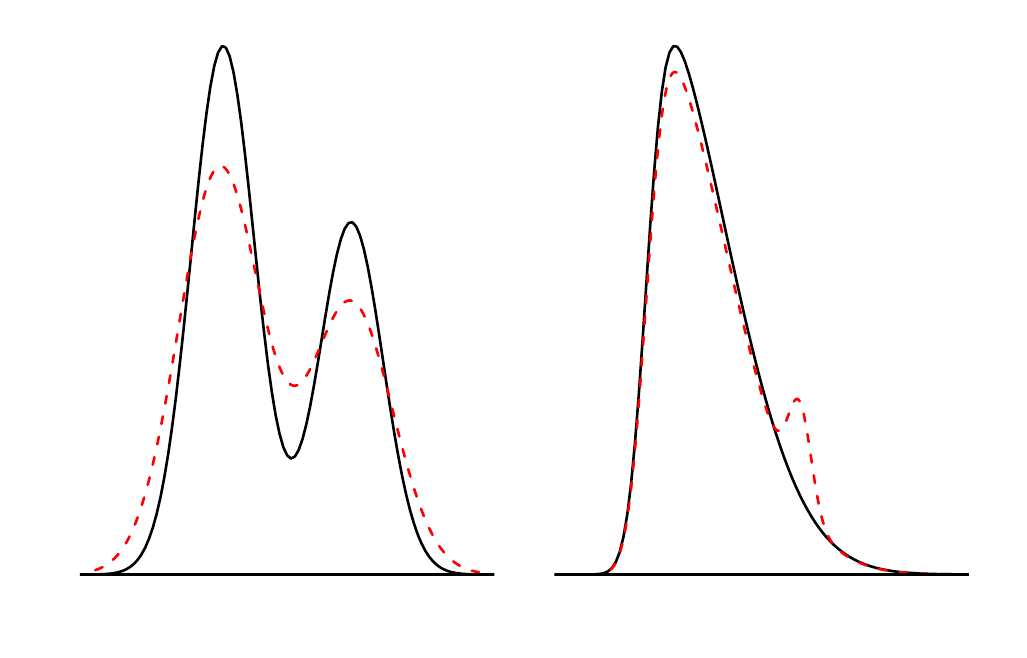}
\caption{Left picture: two quite different densities, from an ISE perspective, inducing the same partition of the space. Right picture: two closer densities having different number of clusters.}
\label{fig:fig1}
\end{figure}

So far, the aim of choosing an amount of smoothing for the specific task of highlighting clustering structures has been scarcely pursued in literature. A related idea, although without particular reference to cluster analysis, has been developed by \cite{samworthwand}, who propose a plug-in type bandwidth selector appropriate for estimation of highest density regions (see also \cite{qiao2018} and \cite{DW2018}). Another related work, more focused on the clustering problem, is the one by \cite{einbeck2011}, where the author suggests to consider the self-coverage measure as a criterion for bandwidth selection. Alternatively, the potential adequacy of a bandwidth selected to properly estimate the density gradient has been pointed out informally by \cite{chaconduong} and explored numerically by \cite{chaconmonfort}. The theoretical motivation of this suggestion lies on the strong dependence of both the population modal clustering and the mean shift updating mechanism on the density gradient. The suggestion in \cite{chen2016} follows the same rationale and the bandwidth is proposed to be selected as a modification of the normal reference rule for density gradient estimation.

To address the problem of bandwidth selection for modal clustering, an appropriate measure of distance should compare the data-based clustering induced by a kernel density estimate with the ideal population one. Stemming from \cite{chacon15}, a natural choice is the \textit{distance in measure}, where the considered measure here is the probability $\mathbb P$ induced by the density $f$. Formally, let $\mathscr{C}=\{C_1,\dots,C_r\}$ and $\mathscr{D}=\{{D}_1,\dots,{D}_s \}$ be two partitions with $r \leq s$ (i.e. possibly different number of groups). The distance in measure between $\mathscr C$ and $\mathscr D$ is defined as
\begin{equation}\label{eq:dm}
d(\mathscr{C},\mathscr{D}) = \frac{1}{2} \min_{\sigma \in \mathcal{P}_s} \left\{ \sum_{i=1}^r \mathbb P(C_i \Delta D_{\sigma(i)}) + \sum_{i=r+1}^s \mathbb P(D_{\sigma(i)}) \right\},
\end{equation}
where $C \Delta D = (C \cap D^c) \cup (C^c \cap D)$ is the symmetric difference between any two sets $C$ and $D$ and $\mathcal{P}_s$ denotes the set of permutations of $\{1,2,\dots,s\}$.
{}
This distance finds an interpretation as the minimal probability mass that would need to be re-labeled to transform one clustering into the other (see Figure \ref{fig:fig2} for a graphical illustration). In this sense, the second term in (\ref{eq:dm}) serves as a penalization for unmatched clusters in one of the clusterings. Practically, this distance conveys the idea that two partitions are similar not when they are physically close, but when the differently-labeled points do not represent a significant portion of the distribution.

\begin{figure}
\centering
\includegraphics[scale=1]{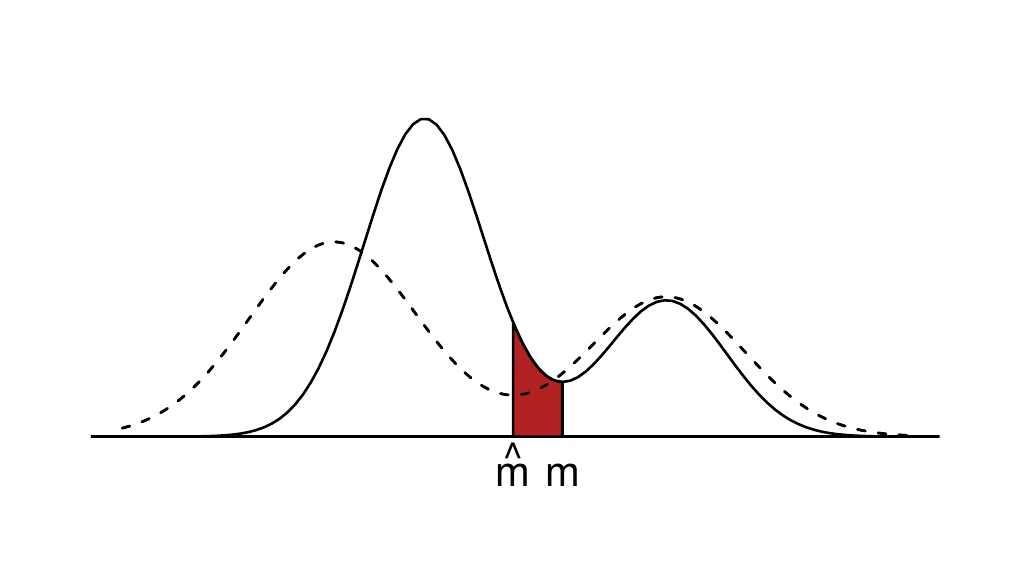}
\caption{Graphical interpretation of the distance in measure: the shaded area represents the probability mass that would need to be re-labeled to transform one induced clustering into the other.}
\label{fig:fig2}
\end{figure}

It should be noted that the choice of this distance to evaluate the performance of a data-based clustering is not arbitrary. Indeed, many other possibilities are described in \cite{meila2016}, but the conclusion of that study is that the distance in measure (called misclassification error there) is ``the distance that comes closest to satifying everyone''. Furthermore, in \cite{vonLuxburg2010} the distance in measure is considered as ``the most convenient choice from a theoretical point of view''.

As with the ISE-MISE duality, the distance in measure is a stochastic error distance, so for the purpose of bandwidth selection it seems more convenient to focus on the \textit{Expected Distance in Measure}
\begin{equation}\label{eq:distmeas}
{\rm EDM}(h)=\mathbb{E}\big[d(\hat{\mathscr C}_h,\mathscr{C}_0) \big],
\end{equation}
where $\hat{\mathscr C}_h$ is the data-based partition induced by $\hat f_h$ and $\mathscr{C}_0$ represents the ideal population clustering. Once the appropriate error distance is defined, the optimal bandwidth $h$ is given by $h_{\rm EDM} = \mathop{\rm argmin}_{h>0} {\rm EDM}(h)$.

As it happened with $h_{\rm MISE}$, it does not seem possible to find an explicit expression for $h_{\rm EDM}$. Hence, our goal will be to obtain an asymptotic form for the EDM that allows to derive a simple approximation to $h_{\rm EDM}$.

To this aim, consider a standard normal random variable $Z$, and denote by $\psi(\mu,\sigma^2)=\mathbb E|\mu+\sigma Z|$ for $\mu\in\mathbb R$ and $\sigma>0$. Since $|\mu+\sigma Z|$ has a folded normal distribution \citep{LNN61}, it follows that $\psi(\mu,\sigma^2)$ can be explicitly expressed as
\begin{eqnarray}\label{eq:psi}
\psi(\mu,\sigma^2)&=&(2/\pi)^{1/2}\sigma e^{-\mu^2/(2\sigma^2)}+\mu\big\{1-2\Phi(-\mu/\sigma)\big\} \\
&=&(2/\pi)^{1/2}\Big\{\sigma e^{-\mu^2/(2\sigma^2)}+|\mu|\int_0^{|\mu|/\sigma}e^{-z^2/2}dz\Big\}  \; , \nonumber
\end{eqnarray}
where $\Phi$ denotes the distribution function of $Z$. This function $\psi$ plays a key role in the asymptotic behavior of the expected distance in measure, as the next result shows (see Appendix \ref{proofs} for a proof).

\begin{theorem}\label{thm:AEDM}
	Assume that $f$ is a bounded Morse function with $r\geq2$ modes and local minima $m_1<\dots<m_{r-1}$, three-times continuously differentiable around each $m_j$, that $\int_{-\infty}^\infty|x|f(x)dx<\infty$, and that the kernel $K$ is supported on $(-1,1)$, has four bounded derivatives and satisfies $\int_{-\infty}^\infty K(x)dx=1$, $\int_{-\infty}^\infty xK(x)dx=0$ and $\mu_2(K)=\int_{-\infty}^\infty x^2K(x)dx<\infty$. Define $R(K^{(1)})=\int_{-\infty}^{\infty} K^{(1)}(x)^2 dx$ and suppose also that $h\equiv h_n$ is such that $h\to0$, $nh^5/\log n\to\infty$ and $(nh^7)^{-1}$ is bounded. Then, ${\rm EDM}(h)$ is asymptotically equivalent to
\begin{equation}\label{eq:AEDM}
\hspace{-.3cm}{\rm AEDM}(h)=\sum_{j=1}^{r-1}\frac{f(m_j)}{f^{(2)}(m_j)}\psi\Big(\tfrac{1}{2}\mu_2(K)f^{(3)}(m_j)h^2,R(K^{(1)})f(m_j)(nh^3)^{-1}\Big),
\end{equation}
where $g^{(k)}$ refers to the $k$-th derivative of a function $g(\cdot)$.
\end{theorem}

The asymptotically optimal bandwidth $h_{\rm AEDM}$ is then defined as the value of $h>0$ that minimizes ${\rm AEDM}(h)$. Due to the structure of $\psi(\cdot,1)$, minimization of (\ref{eq:AEDM}) is closely related to the problem of minimizing the $L_1$ distance in kernel density estimation and, in fact, reasoning as in \cite{HW88} it is possible to show that $h_{\rm AEDM}$ is of order $n^{-1/7}$. Unfortunately, as it happened with $h_{\rm EDM}$, it seems that neither $h_{\rm AEDM}$ admits an explicit representation hence, to get further insight into the problem of optimal bandwidth selection for density clustering, it appears necessary to rely on a tight upper bound for ${\rm AEDM}(h)$.

To find such a bound it is useful to note that many properties of $\psi(u,1)$ are given in \cite[Ch. 5]{DG85}, and can be translated to our function of interest by taking into account that $\psi(\mu,\sigma^2)=\sigma\psi(\mu/\sigma,1)$. It follows that $\psi(\mu,\sigma^2)$ is symmetric with respect to $\mu$, nondecreasing for $\mu>0$ and convex, attaining its minimum at $\mu=0$ so that $\psi(\mu,\sigma^2)\geq\psi(0,\sigma^2)=(2/\pi)^{1/2}\sigma$ for all $\mu\in\mathbb R,\sigma>0$.

By taking into account that $e^{-\mu^2/(2\sigma^2)}$ and $1-2\Phi(-\mu/\sigma)$ are both bounded by 1, \cite{DG85} also noted that
\begin{equation}\label{eq:bound1}
\psi(\mu,\sigma^2)\leq(2/\pi)^{1/2}\sigma+|\mu|
\end{equation}
for all $\mu\in\mathbb R,\sigma>0$. However, a tighter bound for small values of $\mu$ is given in the next lemma.

\begin{lemma}\label{lem:bound2}
The bound $\psi(\mu,\sigma^2)\leq(2/\pi)^{1/2}\sigma+(2\pi)^{-1/2}\mu^2/\sigma$ holds for all $\mu\in\mathbb R$ and $\sigma>0$.
\end{lemma}

The bound in Lemma \ref{lem:bound2} is tighter than (\ref{eq:bound1}) whenever $|\mu|\leq(2\pi)^{1/2}\sigma$, but the situation reverses for bigger values of $|\mu|$, so that none of the two bounds is uniformly better (see Figure \ref{fig:psi-bounds}) hence we should keep track of both of them. They lead to upper bounds for the asymptotic EDM.

\begin{figure}\centering
\includegraphics[width=.5\textwidth]{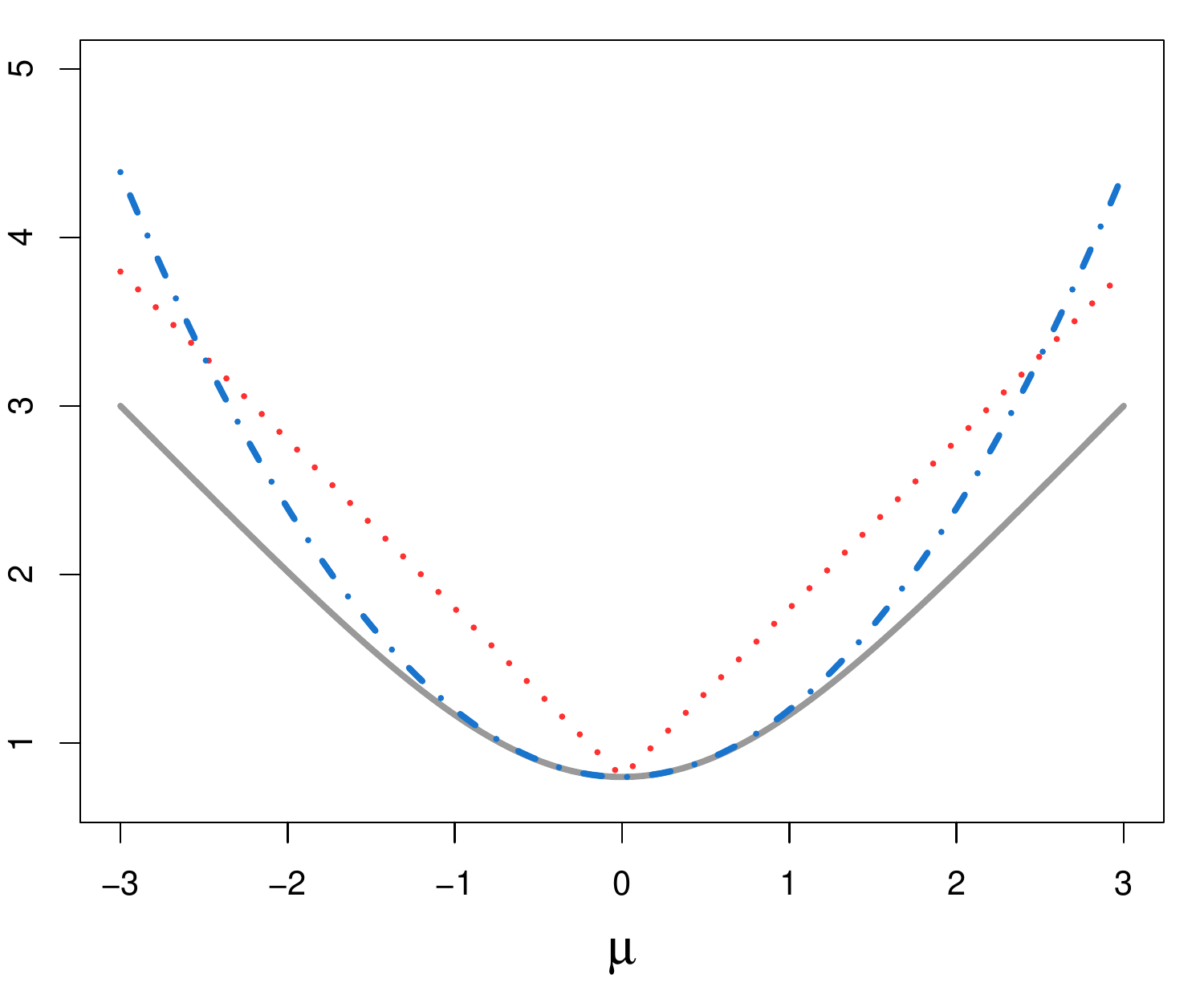}
\caption{Graph of $\psi(\mu,1)$ as a function of $\mu$ (grey solid curve), together the bound (\ref{eq:bound1}) (red dotted line) and the bound from Lemma \ref{lem:bound2} (blue dot-dashed curve).}
\label{fig:psi-bounds}
\end{figure}

\begin{corollary}\label{cor:1}
Under the conditions of Theorem \ref{thm:AEDM}, the asymptotic EDM satisfies ${\rm AEDM}(h)\leq\min\{{\rm AB}1(h),{\rm AB}2(h)\}$ for all $h>0$, where
\begin{align*}
{\rm AB}1(h)&=(2/\pi)^{1/2}R(K^{(1)})^{1/2}bn^{-1/2}h^{-3/2}+\tfrac12\mu_2(K)a_1h^2,\\
{\rm AB}2(h)&=(2/\pi)^{1/2}R(K^{(1)})^{1/2}bn^{-1/2}h^{-3/2}+\\ &\hspace{.5cm}+(32\pi)^{-1/2}\mu_2(K)^2R(K^{(1)})^{-1/2}a_2n^{1/2}h^{11/2}.
\end{align*}
Here, $b=\sum_{j=1}^{r-1}b_{j}$ and $a_\ell=\sum_{j=1}^{r-1}a_{j\ell}$ and for $\ell=1,2$, where
\begin{align*}
a_{j1}&=f(m_j)|f^{(3)}(m_j)|/f^{(2)}(m_j),&b_{j}=f(m_j)^{3/2}/f^{(2)}(m_j),\\
a_{j2}&=f(m_j)^{1/2}f^{(3)}(m_j)^2/f^{(2)}(m_j).
\end{align*}
The minimizers of ${\rm AB}1(h)$ and ${\rm AB}2(h)$ can be computed explicitly, and are given by
\begin{align}\label{eq:hbound1}
      	h_{\rm AB1} &= \left( \frac{9R(K^{(1)})b^2} {2\pi\mu_2(K)^2a_1^2} \right)^{1/7}n^{-1/7} \\\label{eq:hbound2}
		h_{\rm AB2} &=  \left( \frac{24R(K^{(1)}) b }{11\mu_2(K)^2 a_2} \right)^{1/7} n^{-1/7} \; .
    	\end{align}
\end{corollary}


\subsection{Some remarks}\label{sec:discussion}
In this section we discuss in more depth some of the results derived in Section \ref{sec:bwselect}. The aim is to provide insights on the behavior of the approximations and bandwidth selectors and to discuss possible competitors.

\begin{remark}\label{rem:r1}
Theorem \ref{thm:AEDM} provides an asymptotic expression for the EDM that is valid as long as the true density has two or more modes. When the true density is unimodal ($r=1$), expression \eqref{eq:AEDM} is not well-defined. However, under the assumptions of the theorem the kernel estimator is also unimodal with probability one for big enough $n$. Thus, asymptotically the distance in measure would be identically zero, hence the AEDM formula would remain valid under the usual convention setting $\sum_{j=1}^0=0$.

Moreover, for unimodal densities the numerical work in Section \ref{sec:numresults} suggests that there exists $h_0>0$ such that ${\rm EDM}(h)=0$ for all $h\geq h_0$. Hence, in that case it seems sensible to define $h_{\rm EDM}=\inf\{h>0\colon {\rm EDM}(h)=0\}$.


\end{remark}

\begin{remark}
A natural estimator of the density first derivative is the first derivative of the kernel density estimator. For this estimator it is possible to define the MISE as in (\ref{eq:MISE}), and to consider its minimizer $h_{\rm MISE,1}$ and its asymptotic approximation $h_{\rm AMISE,1}$ (see \cite{singh87} and \cite{chacon_etal11}). The bandwidths (\ref{eq:hbound1}) and (\ref{eq:hbound2}) share the same order as $h_{\rm AMISE,1}$, whose expression is given by
	\begin{equation}\label{eq:hgrad}
 		h_{\rm AMISE,1}=\left( \frac{3R(K^{(1)})}{\mu_2(K)^2R(f^{(3)})} \right)^{1/7} n^{-1/7},
 	\end{equation}
 with $R(f^{(3)})=\int_{-\infty}^\infty f^{(3)}(x)^2dx$.
 	This consideration strengthens the intuition, outlined in Section \ref{sec:bwselect}, that (\ref{eq:hgrad}) could be an adequate bandwidth choice for modal clustering purposes.

\end{remark}

\begin{remark}
By explicitly plugging expression (\ref{eq:psi}) for $\psi$ into (\ref{eq:AEDM}), it is easily seen that the AEDM can be decomposed into two summands. Studying their behavior, as a function of $h$, it can be checked that the first term decreases when $h\to 0$ while the second one tends to increase, and viceversa for $h$ taking large values. A similar trade-off occurs with the decomposition of the AMISE into the \textit{Asymptotic Integrated Squared Bias} and the \textit{Asymptotic Integrated Variance}, which are minimized for diverging values of $h$.
\end{remark}

\begin{remark}
If the true density is exactly symmetric around its minimum, the considerations in the previous item do not hold anymore. Symmetry around a minimum $m$ implies $f^{(k)}(m)=0$, for any odd value of $k$. Therefore the first summand of the AEDM expression, related to the bias, vanishes, leading to a monotonically decreasing behavior of the AEDM itself. A similar anomaly was observed in the related problem of mode estimation in \cite{chernoff1964}: if the true density is symmetric around its mode, then Chernoff's mode estimator is unbiased. Hence, in some special cases symmetry plays a certain role in the performance of these smoothing methodologies.
\end{remark}

\begin{remark}
The derived bandwidths depend on some unknown quantities such as the true density, its local minima and its second and third derivatives. In order to be of practical use we shall resort to plug-in strategies, that is, data-based bandwidth selectors will be proposed in the next section by substituting the aforementioned unknown quantities with pilot estimates. This is the same procedure that is commonly adopted when considering the plug-in bandwidth selector $\hat h_{\rm PI,1}$ for density gradient estimation (see \cite{jones1992} and \cite{chaconduong}).

It should be noted that this allows, from a practical point of view, to circumvent the issue about the perfect symmetry around a minimum since, by using a nonparametric pilot estimate of the third derivative, it is highly unlikely to encounter a similar situation in practice.

\end{remark}

\begin{figure}[t!]\centering
\hspace{-.45cm}	\includegraphics[scale=0.3]{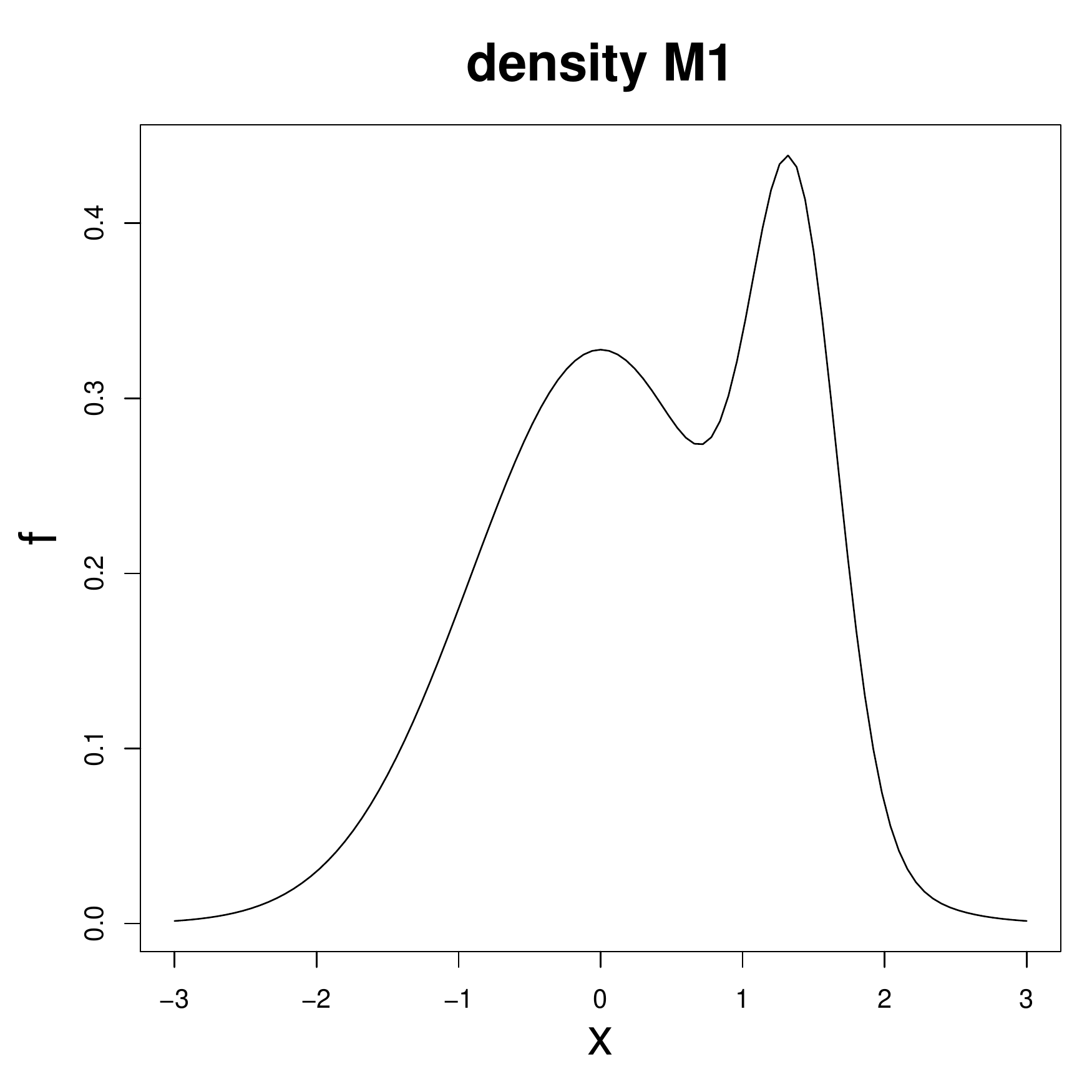}\hspace{-.25cm}
	\includegraphics[scale=0.3]{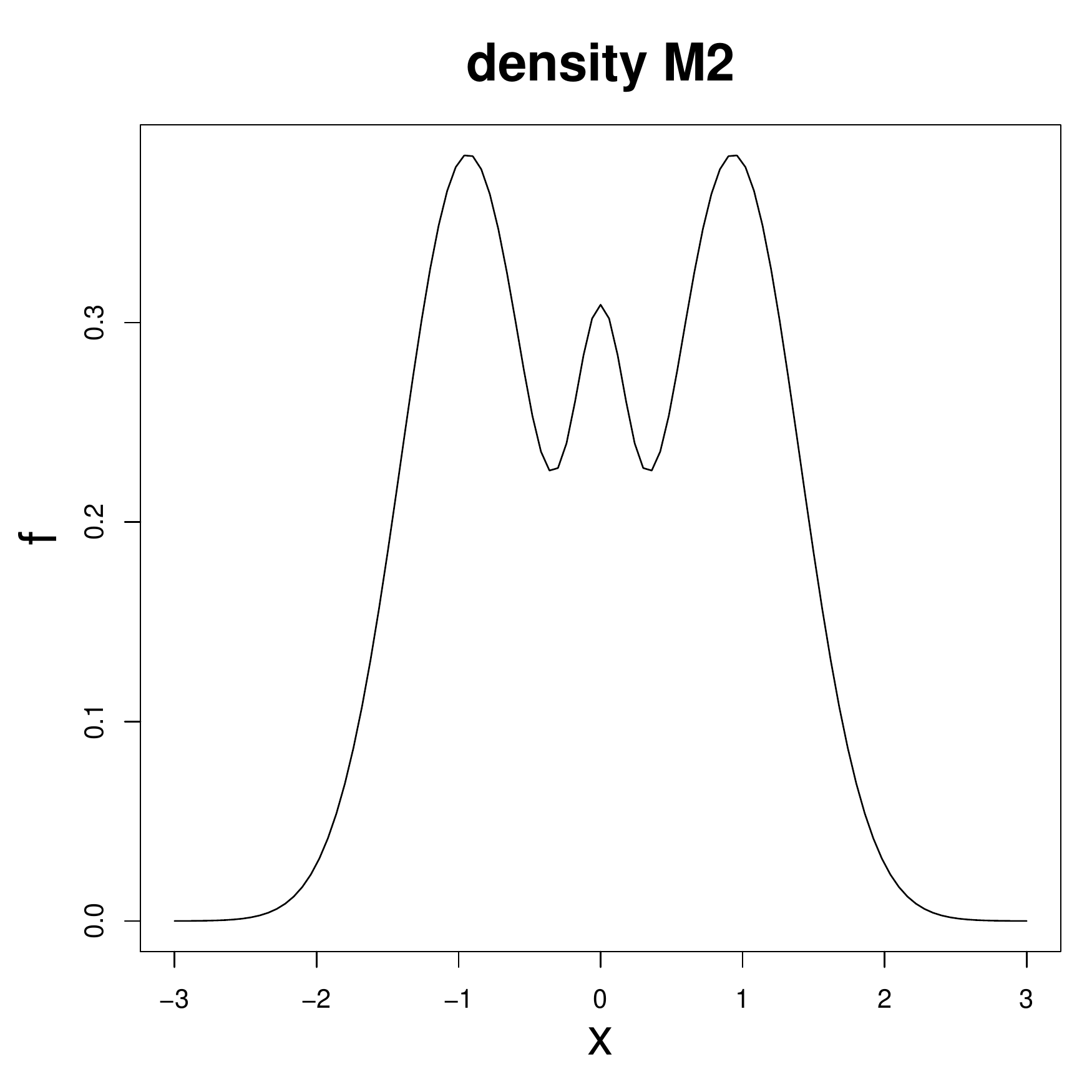}\hspace{-.25cm}
	\includegraphics[scale=0.3]{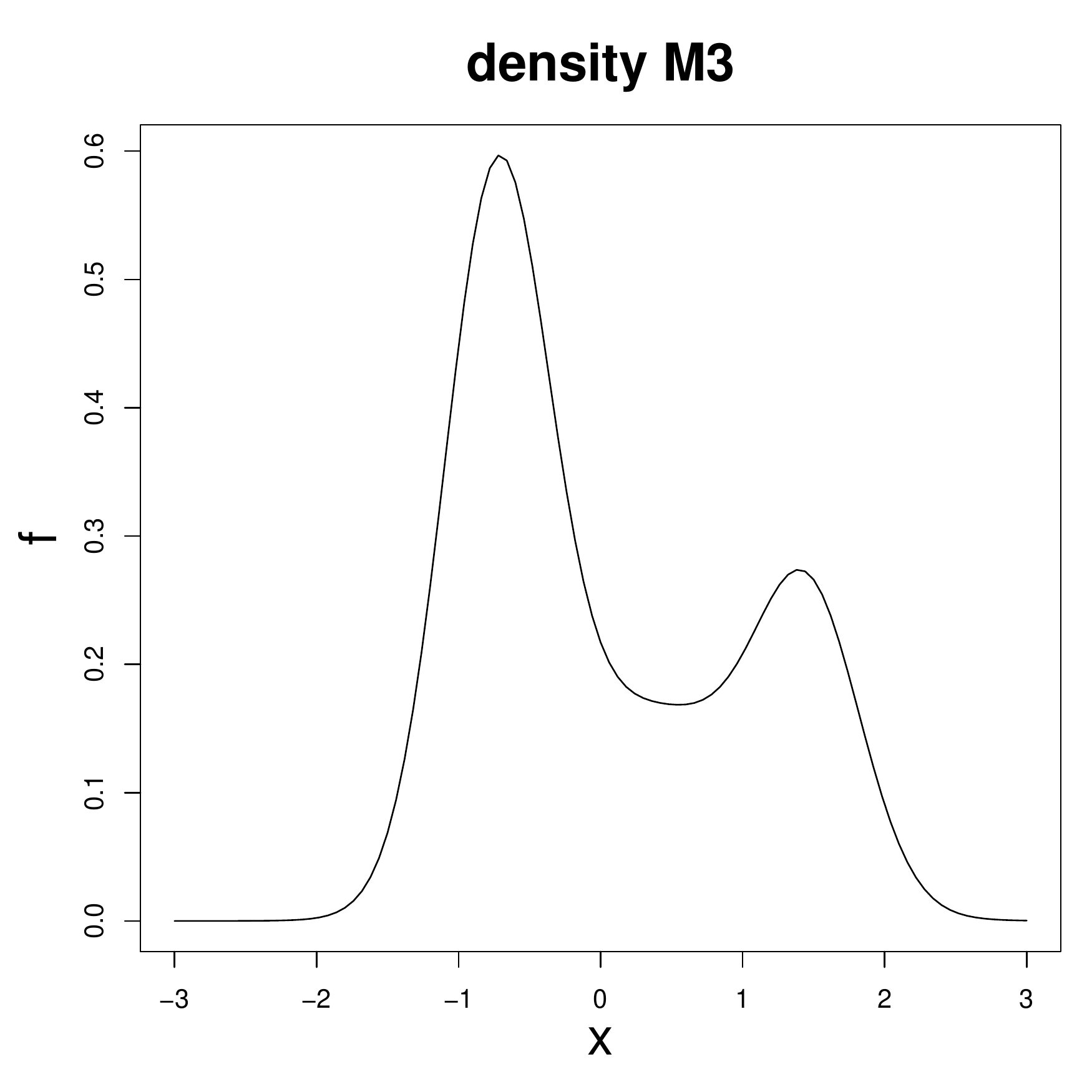}\\
\hspace{-4.35cm}\includegraphics[scale=0.3]{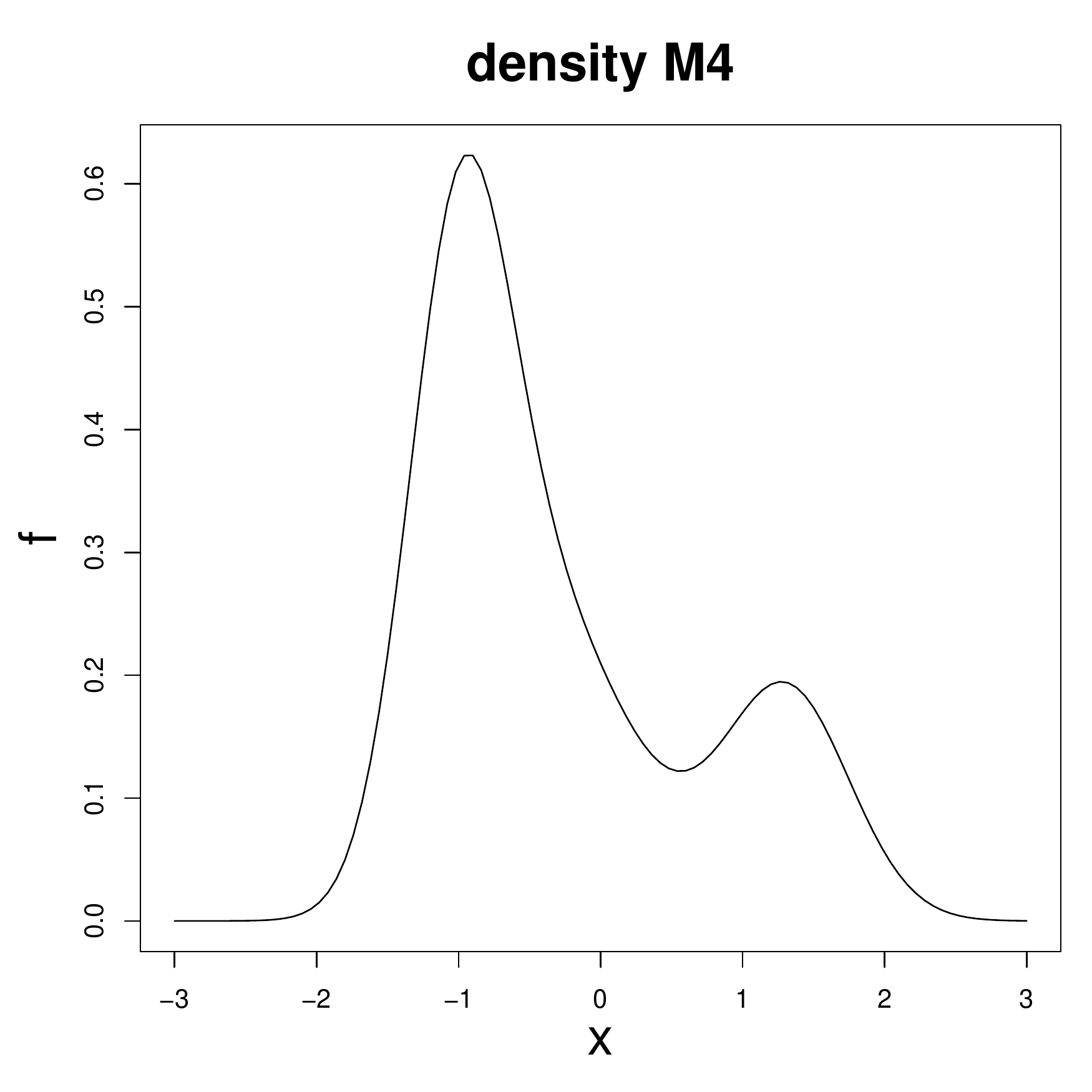}\hspace{-.25cm}
\includegraphics[scale=0.3]{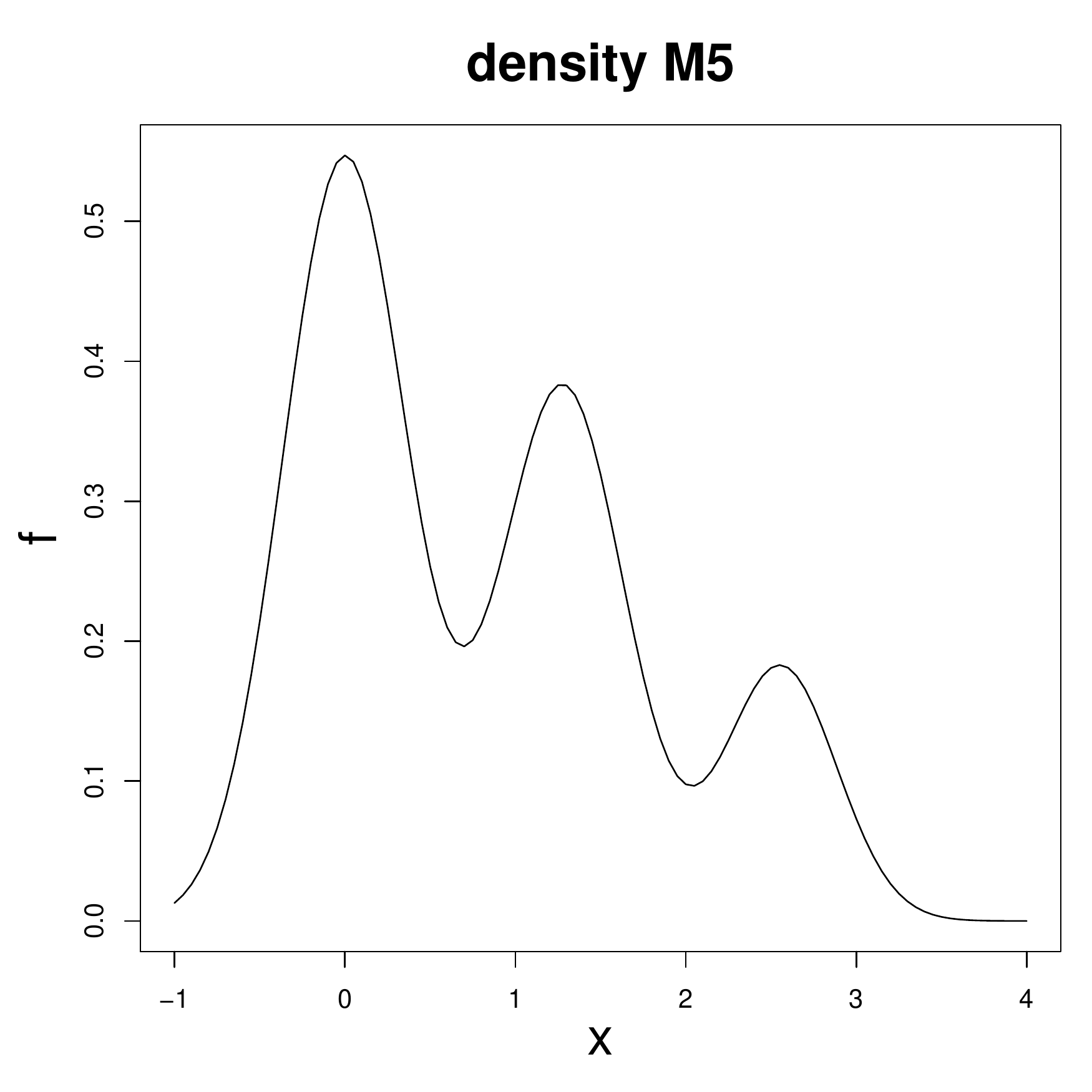}
\caption{Univariate density functions selected for simulations.}\label{fig:densities1}
\end{figure}

\section{Numerical results}\label{sec:numresults}

The idea of estimating the density for clustering purposes, via the minimization of the expected distance in measure -- or its asymptotic counterpart -- is explored in this section via simulations. 
All the analyses have been performed in the R environment \cite{Rsoftware} with the aid of the \texttt{ks}  \cite{kspackage}, \texttt{meanShiftR} \cite{meanshiftr}, \texttt{clue} \cite{clue}, and \texttt{multimode} \cite{multimode} packages.

A total of $B=1000$ samples for each of the sizes $n \in \{100, 1000, 10000\}$ are generated from the univariate densities depicted in Figure \ref{fig:densities1} and whose parameters are reported in Appendix \ref{App:settings}. The selected densities are designed to illustrate different modal structures to encompass different possible behaviors from a clustering perspective.

The first goal of the study was to evaluate the quality of the asymptotic approximation of the EDM and the behavior of the two bounds derived in Corollary \ref{cor:1}. Since an explicit expression for the EDM was not available, we obtained a Monte Carlo approximation based on the $B=1000$ synthetic samples.

The plots displayed in Tables \ref{tab:density8} to \ref{tab:trimodal} show the behavior of the asymptotic approximations, with respect to the EDM, as a function of the bandwidth $h$. As expected, the approximations improve as the sample size increases. The two bounds show a quite different behavior, with characteristics that reflect the theoretical properties pointed out in Section \ref{sec:discussion}. The first bound is closer to the AEDM in uniform terms, but despite having a diverging behavior for large $h$ the second bound is usually closer to the AEDM around the location of the minimizer $h_{\rm AEDM}$.

With regard to the EDM, it presents a nearly flat pattern around its minimizer, thus suggesting a range of plausible bandwidths with very similar performance as the optimal one. This is especially true for densities with a simpler modal structure, captured by the kernel estimate for a wide range of bandwidth values, which suggests that bandwidth selection for modal clustering is, in fact, easier than for density estimation.

\begin{table}[t!]
\caption{Top panel: the EDM (solid line), the AEDM (dashed grey line), and the bounds
AB1 (dotted line) and AB2 (dot-dashed line) versus $h$, for $n =100, 1000, 10000$. All the expressions are evaluated by assuming $f$ and all the involved quantities known. The minimum EDM is reported below the plots, together with the EDM for the oracle bandwidths $h_{\rm AEDM}$ and $h_{\rm MISE,1}$. Bottom panel: average distances in measure (and their standard error) for the proposed bandwidth selectors and the plug-in bandwidth for density gradient estimation. Results refer to density M1.}
\label{tab:density8}
\begin{center}
\begin{tabular}{lccc}
\hline
 &
n =100&n=1000&n=10000\\
  \hline
& \hspace{-.5cm}  \includegraphics[height=2.5cm]{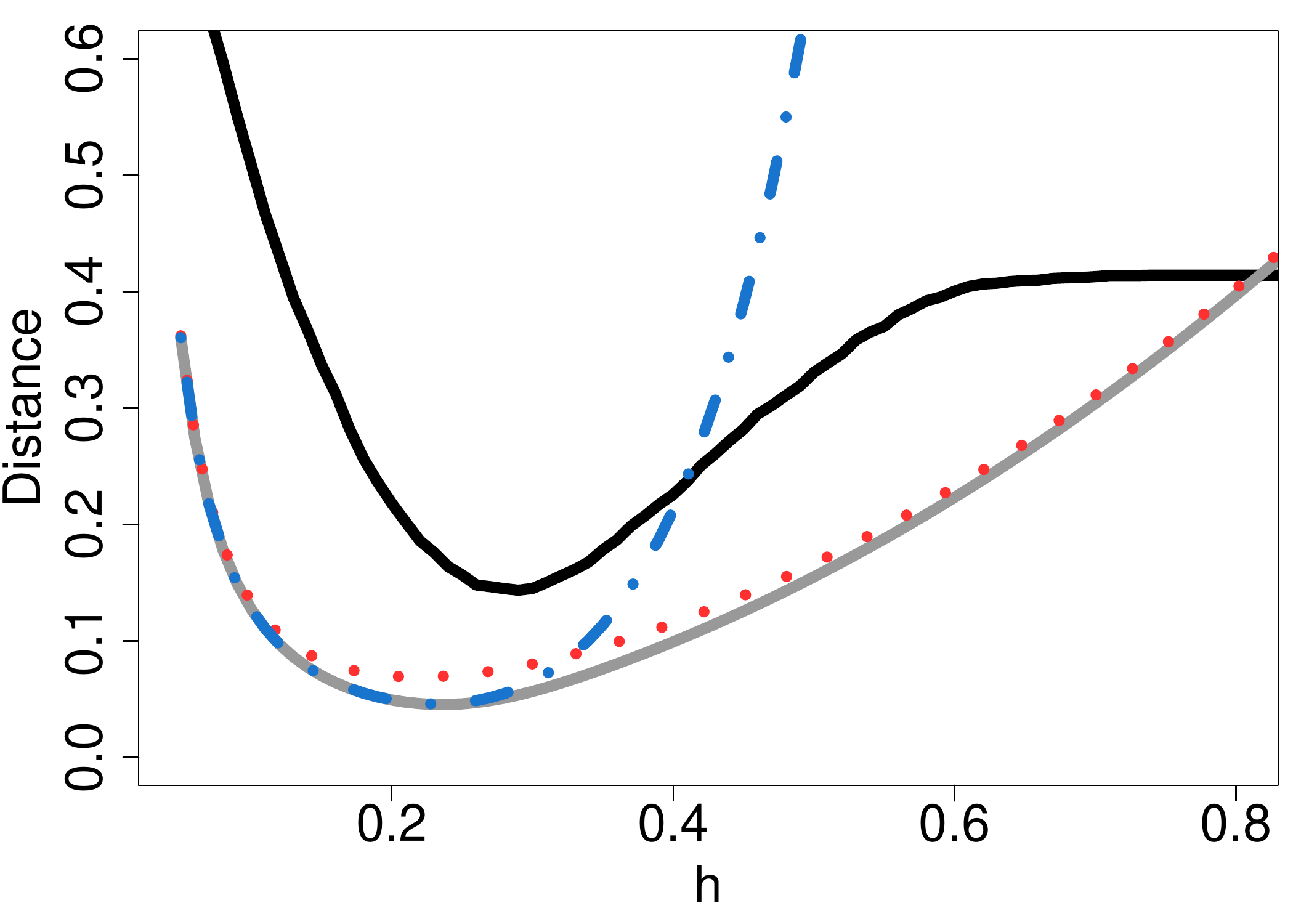}&
   \hspace{-.5cm} \includegraphics[height=2.5cm]{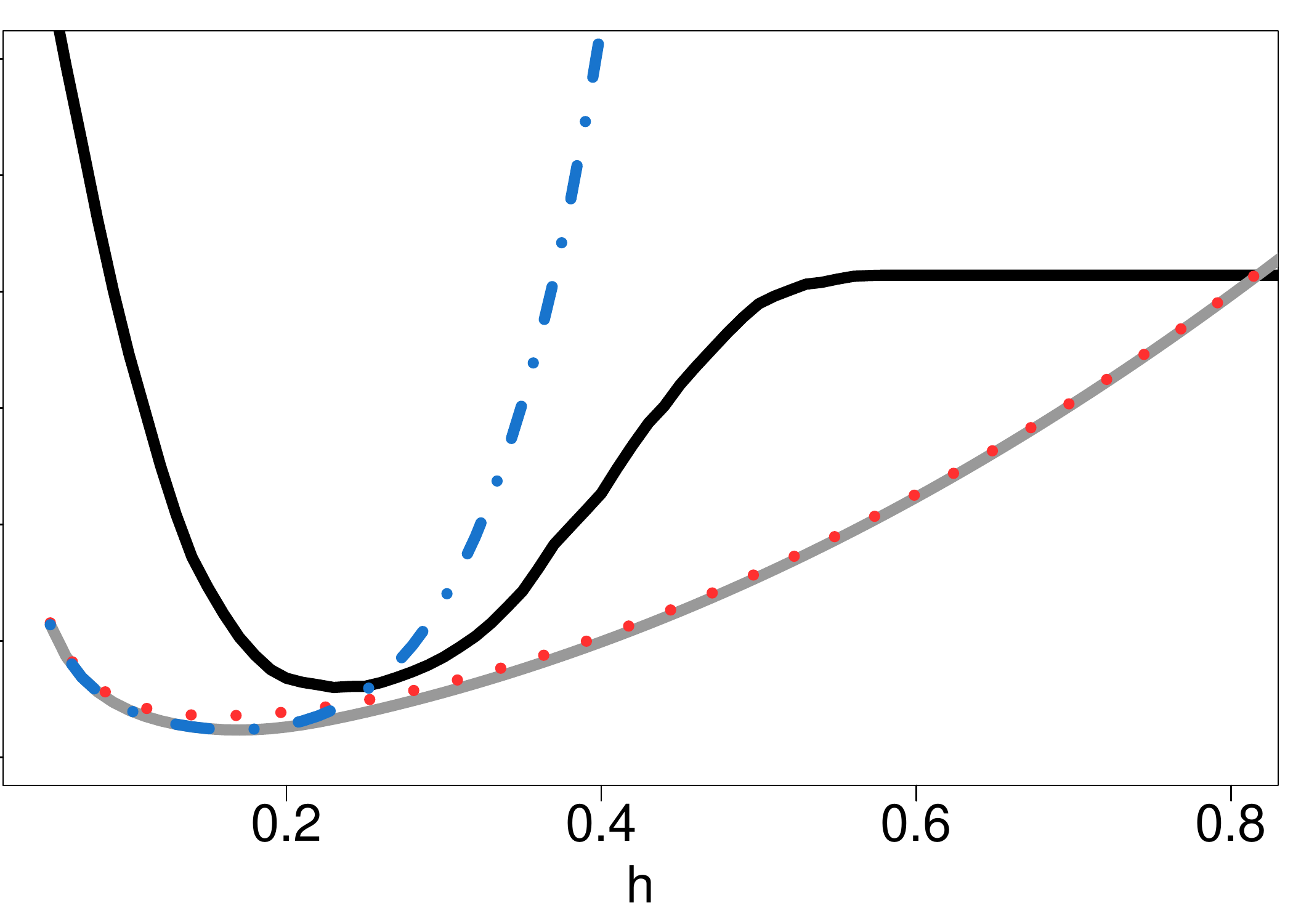}&
   \hspace{-.5cm}  \includegraphics[height=2.5cm]{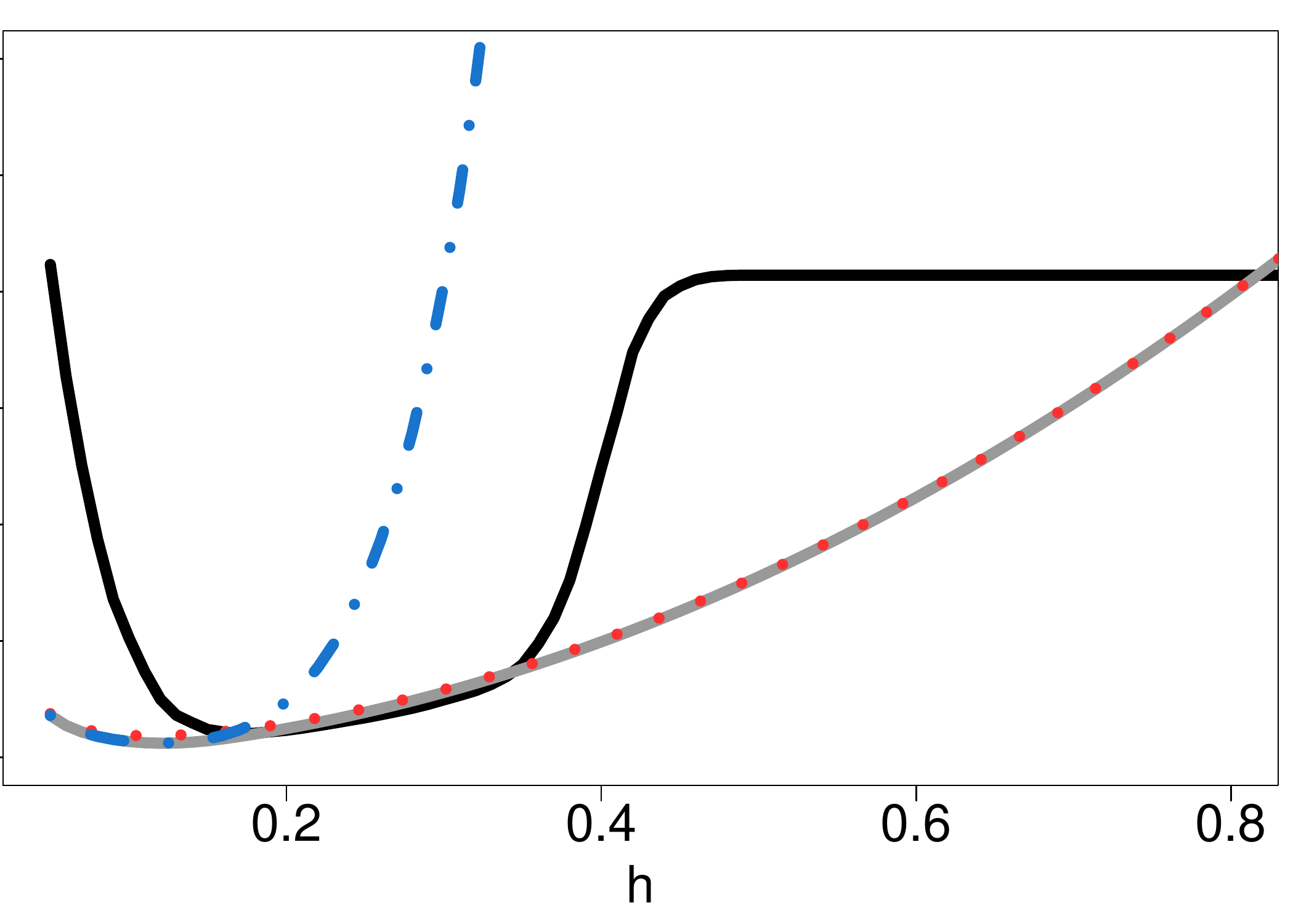}\\
 $h_{\rm EDM}$ & 0.144 & 0.060 & 0.020 \\	
 $h_{\rm AEDM}$ & 0.164 & 0.103 & 0.050 \\
 $h_{\rm MISE, 1} $ & 0.146 & 0.081 & 0.044 \\
   \hline
 $\hat{h}_{\rm AEDM}$ & 0.331 (0.148) &0.150 (0.164) &0.034 (0.076)\\
 $\hat{h}_{\rm AB1}$ &0.297 (0.161) &0.111 (0.134) &0.027 (0.060)\\
  $\hat{h}_{\rm AB2}$ & 0.320 (0.150) &0.128 (0.147) &0.028 (0.060)\\
   $\hat{h}_{\rm PI, 1}$ & 0.221 (0.176) &0.063 (0.084) & 0.029 (0.052) \\
 \hline
\end{tabular}
\end{center}
\end{table}

\begin{table}[t!]
\caption{Cf. Table \ref{tab:density8}. Results refer to density M2.}\label{tab:density9}
\begin{center}
\begin{tabular}{lccc}
\hline
 &
n =100&n=1000&n=10000\\
  \hline
& \hspace{-.5cm}  \includegraphics[height=2.5cm]{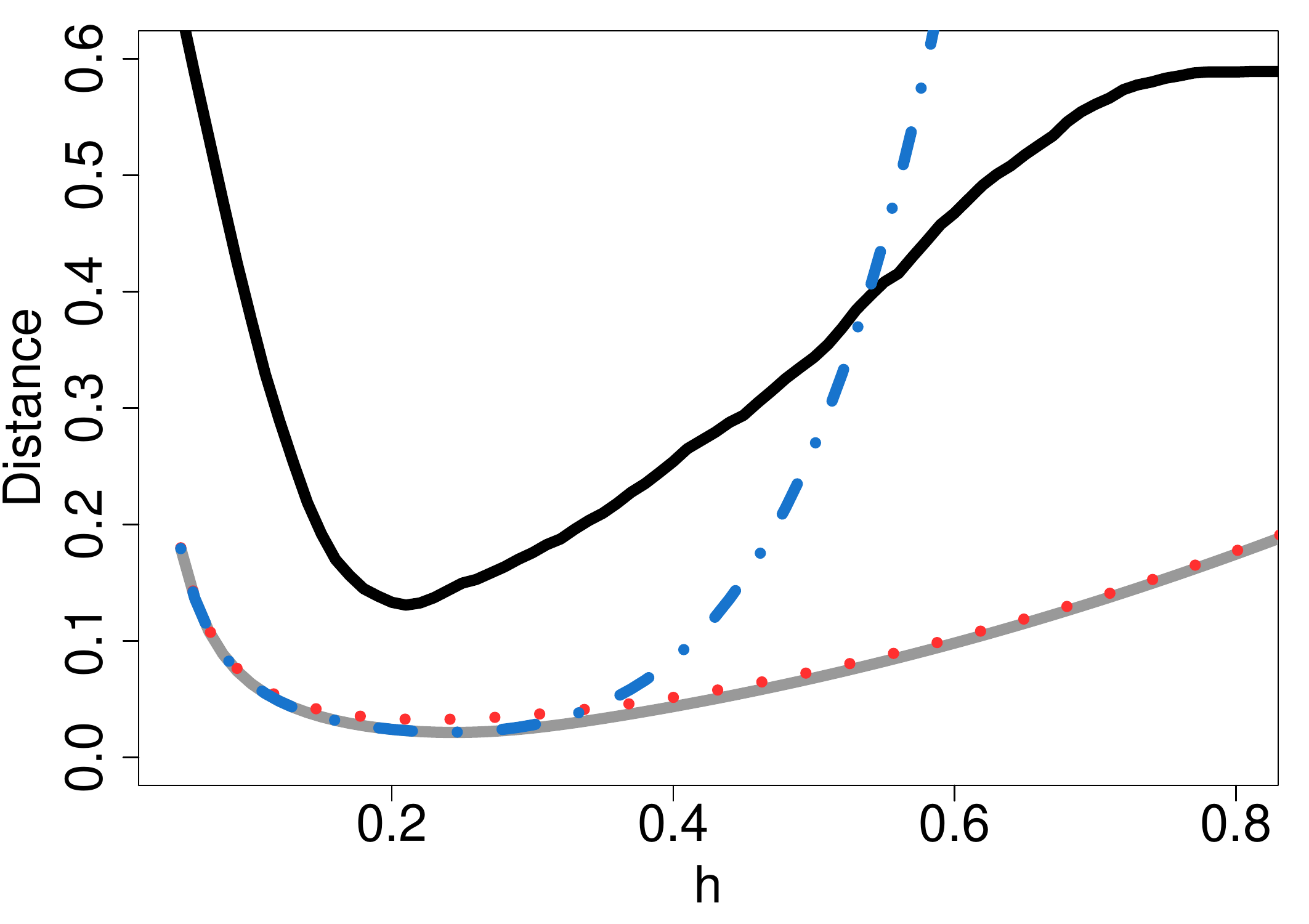}&
   \hspace{-.5cm}  \includegraphics[height=2.5cm]{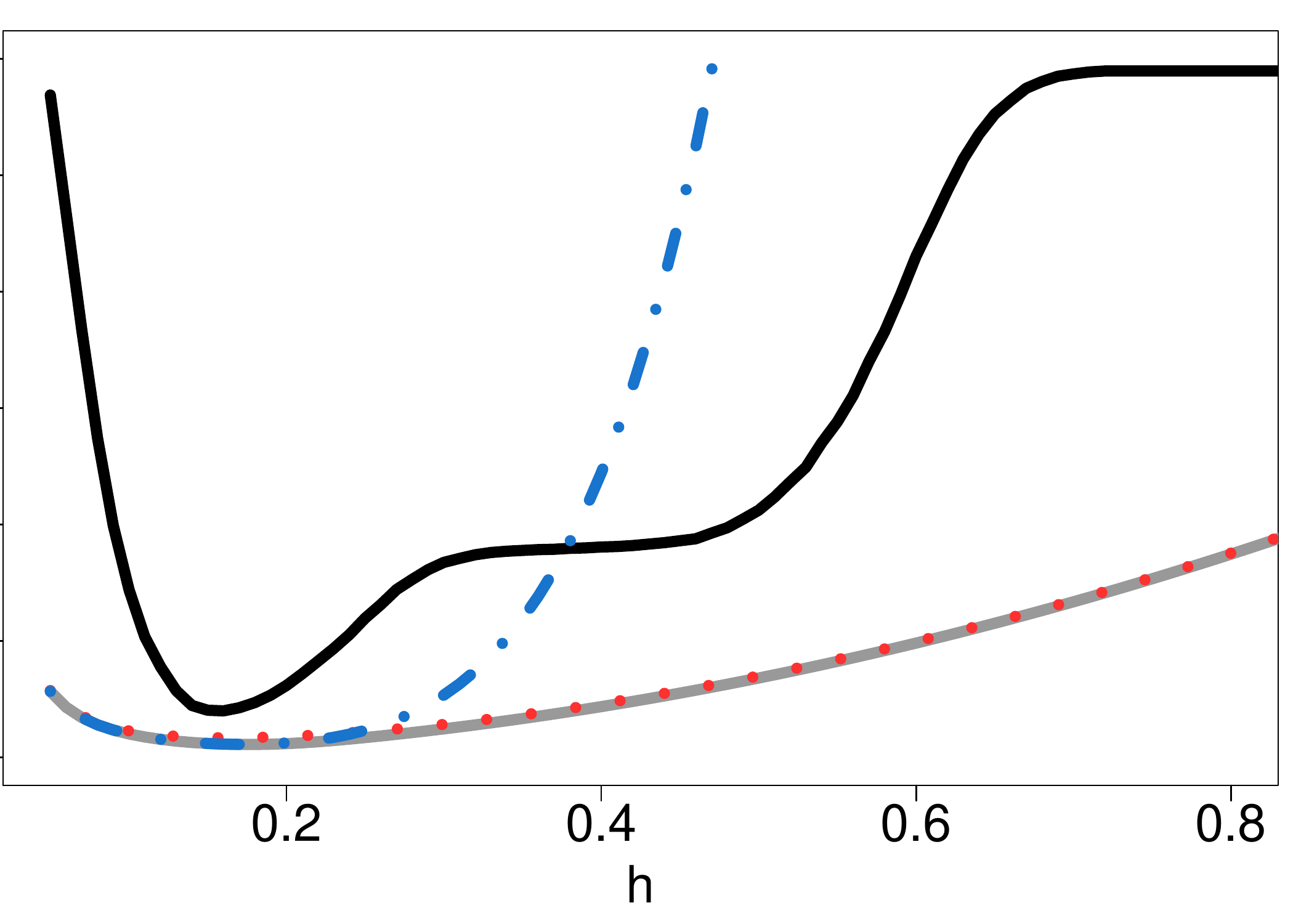}&
   \hspace{-.5cm}  \includegraphics[height=2.5cm]{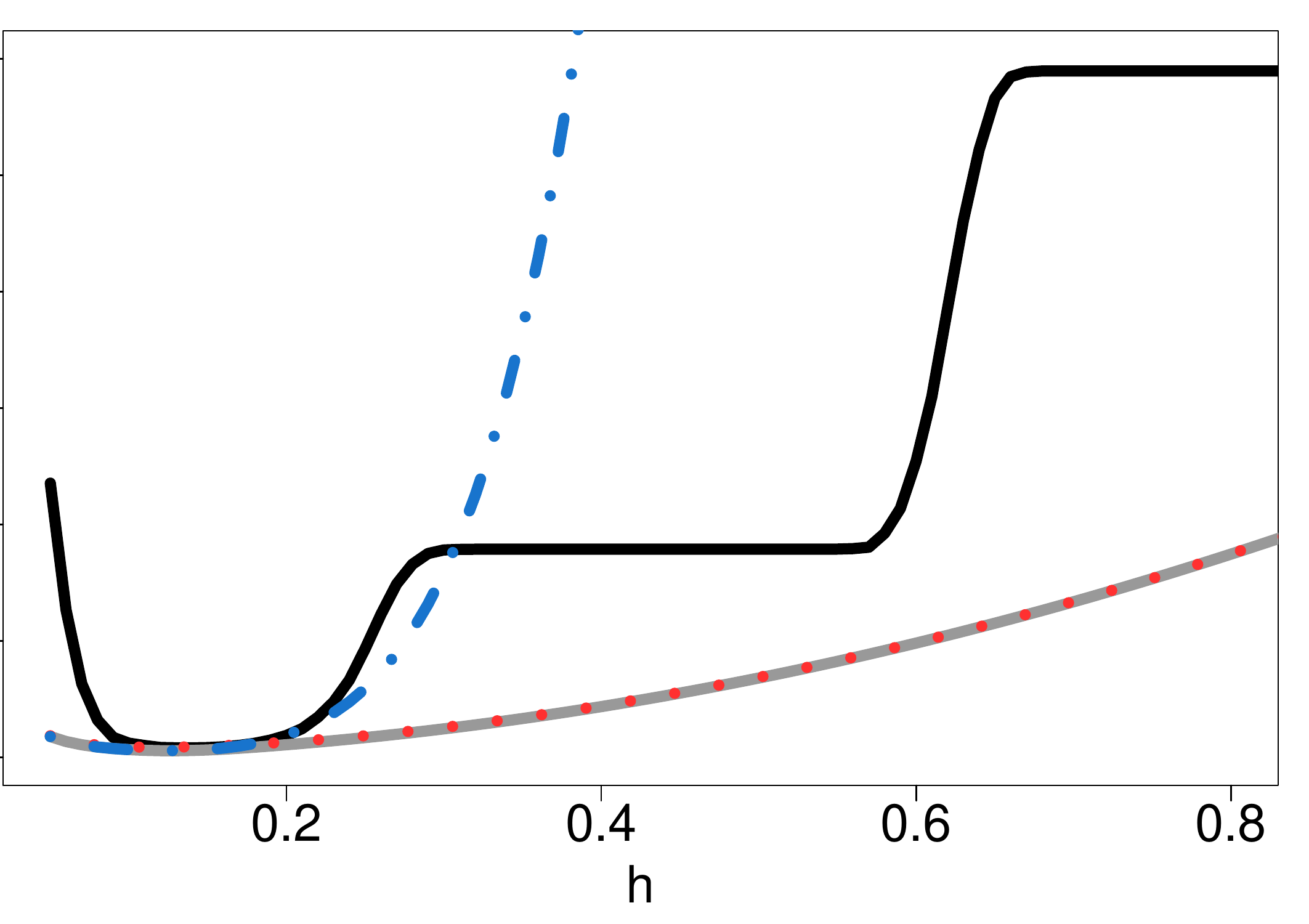}\\
 $h_{\rm EDM}$ & 0.131 & 0.040 & 0.008 \\	
 $h_{\rm AEDM}$ &  0.143 & 0.047 & 0.008 \\
  $h_{\rm MISE, 1} $& 0.165 & 0.041 & 0.011 \\
   \hline
 $\hat{h}_{\rm AEDM}$ & 0.437 (0.195) & 0.201 (0.122) & 0.009 (0.005)\\
 $\hat{h}_{\rm AB1}$ & 0.382 (0.201) & 0.185 (0.118) & 0.008 (0.005) \\
  $\hat{h}_{\rm AB2}$ & 0.423 (0.197) & 0.196 (0.118) & 0.008 (0.005) \\
   $\hat{h}_{\rm PI, 1}$ & 0.256 (0.159) & 0.092 (0.076) & 0.008 (0.005) \\
 \hline
\end{tabular}
\end{center}
\end{table}

\begin{table}[H]
\caption{Cf. Table \ref{tab:density8}. Results refer to density M3.}\label{tab:plateaux}
\begin{center}
\begin{tabular}{lccc}
\hline
 &
n =100&n=1000&n=10000\\
  \hline
& \hspace{-.5cm}  \includegraphics[height=2.5cm]{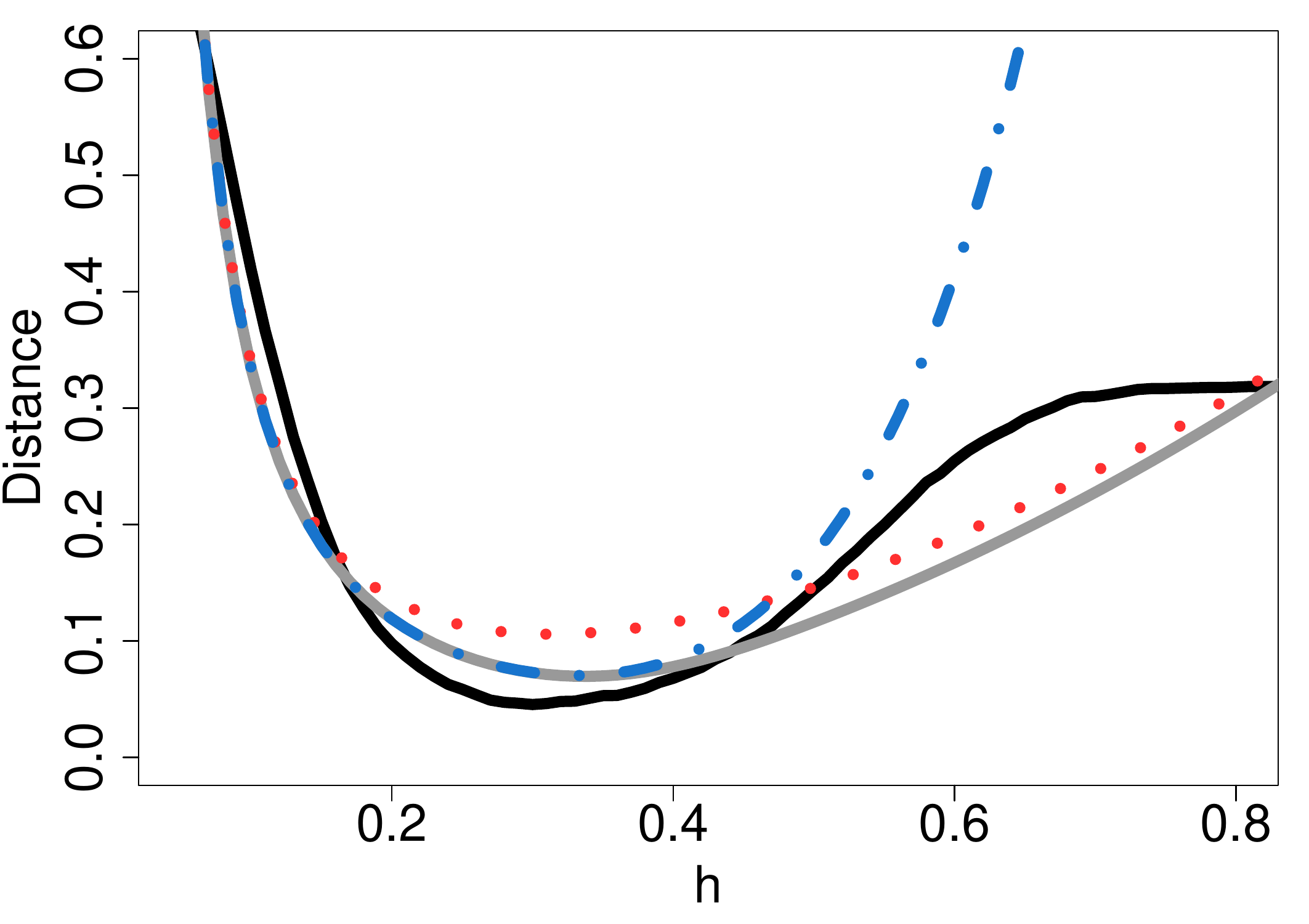}&
   \hspace{-.5cm}  \includegraphics[height=2.5cm]{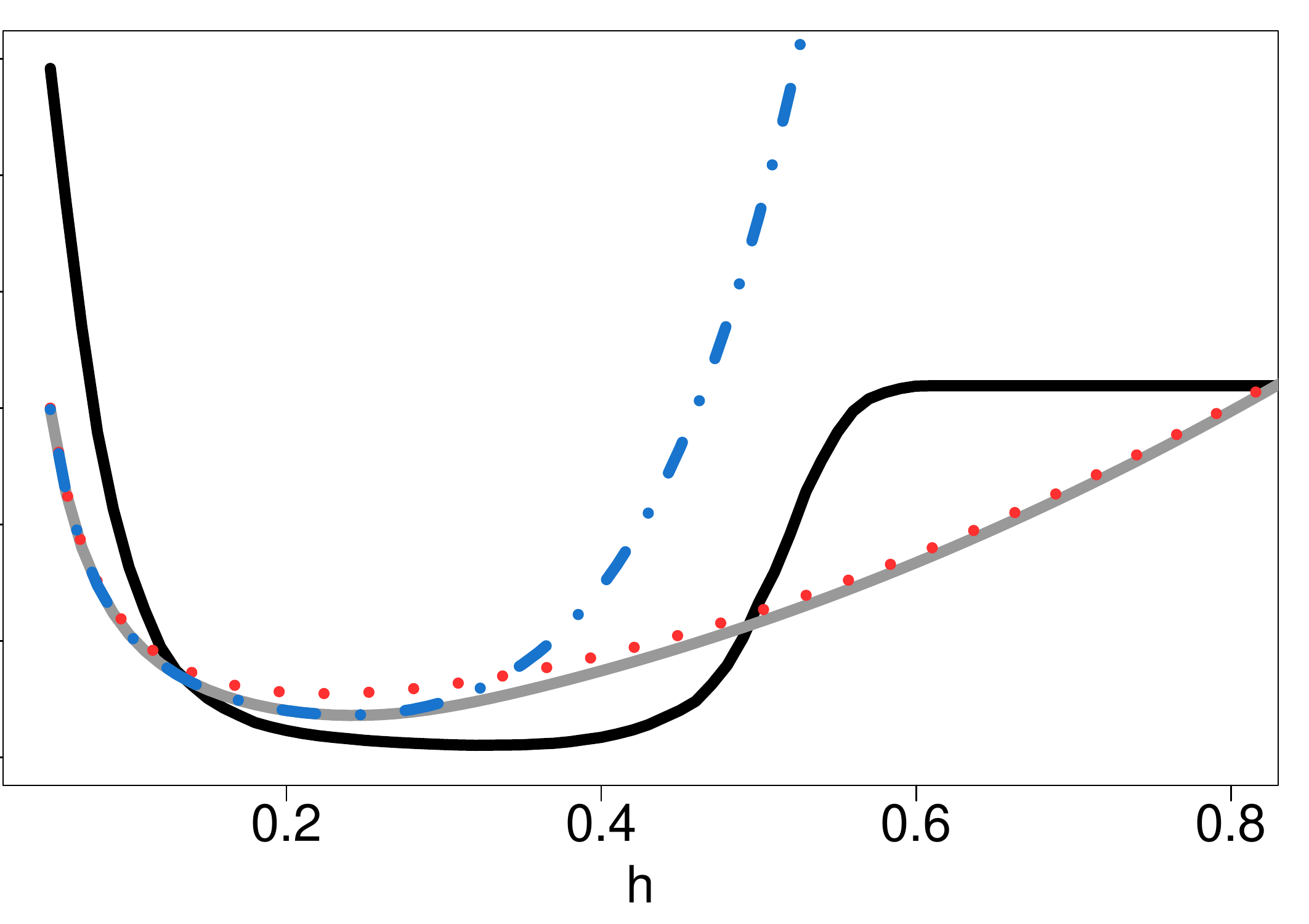}&
   \hspace{-.5cm}  \includegraphics[height=2.5cm]{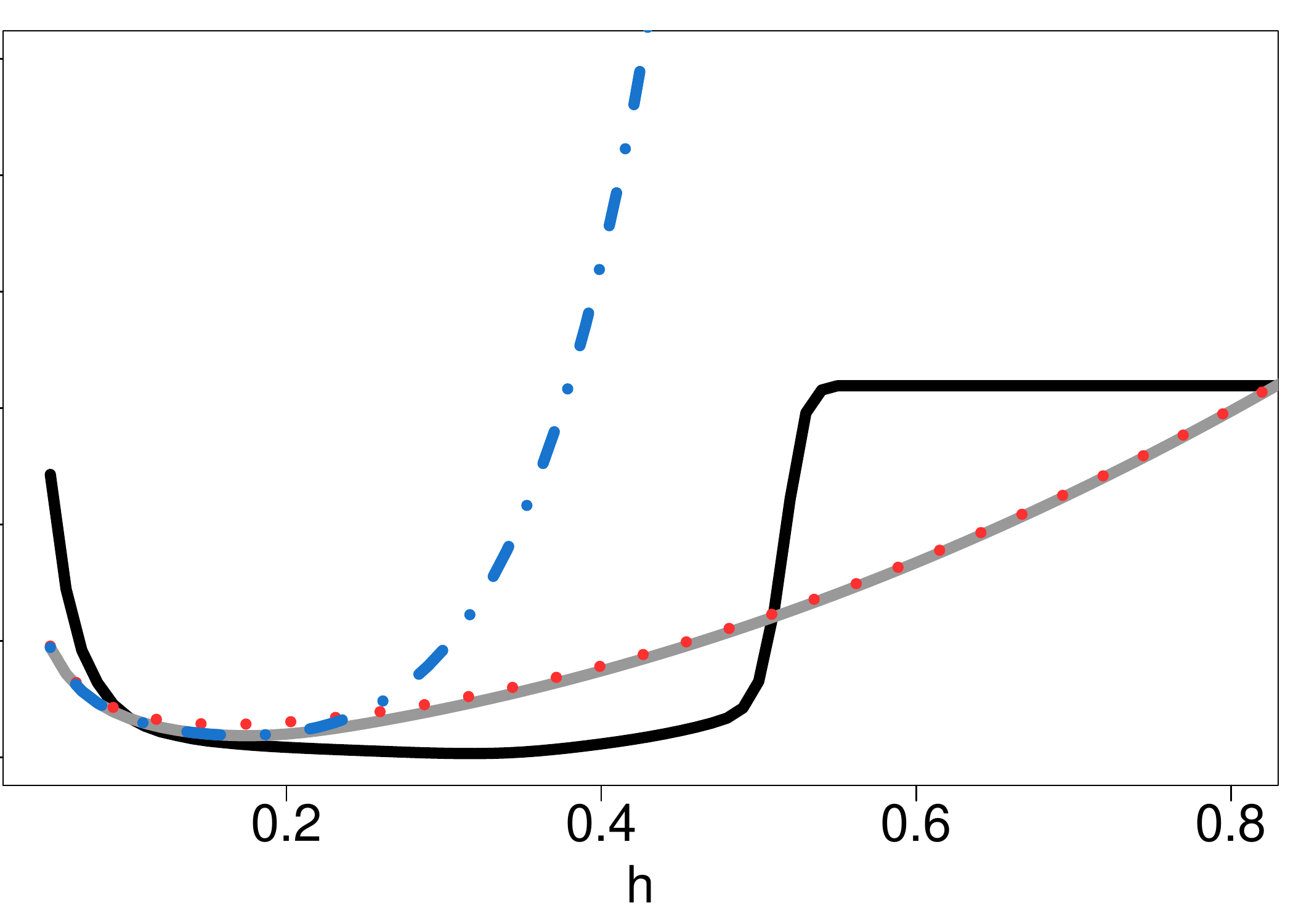}\\
 $h_{\rm EDM}$ & 0.045 & 0.010 & 0.003 \\	
 $h_{\rm AEDM}$ & 0.051 & 0.016 & 0.011 \\
  $h_{\rm MISE, 1} $& 0.054 & 0.034 & 0.022 \\
   \hline
 $\hat{h}_{\rm AEDM}$ & 0.071 (0.072) & 0.042 (0.086) & 0.017 (0.036)\\
 $\hat{h}_{\rm AB1}$ & 0.059 (0.059) & 0.038 (0.075) & 0.017 (0.031) \\
  $\hat{h}_{\rm AB2}$ & 0.067 (0.067) & 0.040 (0.080) & 0.017 (0.033) \\
   $\hat{h}_{\rm PI, 1}$ & 0.050 (0.072) & 0.024 (0.025) & 0.019 (0.017) \\
 \hline
\end{tabular}
\end{center}
\end{table}

To appreciate how much is lost by changing the target from the optimal $h_{\rm EDM}$ to the oracle surrogates $h_{\rm AEDM}$ and $h_{\rm MISE,1}$, the first three lines in each table also present the values for the corresponding EDM, all computed under a full knowledge of the density and its involved features. By construction, ${\rm EDM}(h_{\rm EDM})$ is the lowest of these values and, being derived as an asymptotic approximation, the oracle $h_{\rm AEDM}$ stands close to this optimal value, especially for larger sample sizes. However, it is remarkable that $h_{\rm MISE,1}$, despite being based on a different optimality criterion, also leads to comparable or even improved results over $h_{\rm AEDM}$ in terms of the EDM.

\begin{table}[t!]
\caption{Cf. Table \ref{tab:density8}. Results refer to density M4.}\label{tab:skbim}
\begin{center}
\begin{tabular}{lccc}
\hline
 &
n =100&n=1000&n=10000\\
  \hline
& \hspace{-.5cm}  \includegraphics[height=2.5cm]{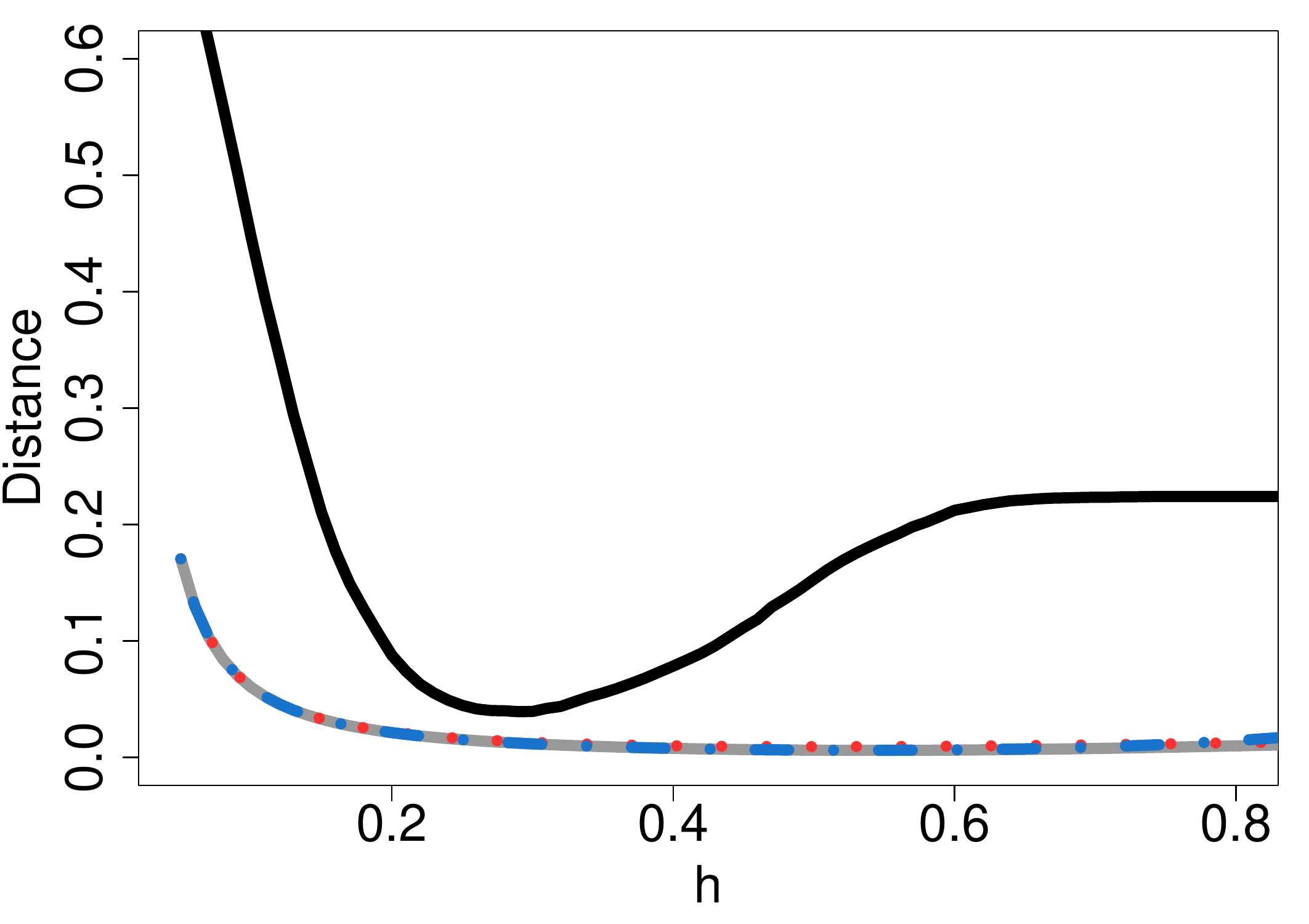}&
   \hspace{-.5cm}  \includegraphics[height=2.5cm]{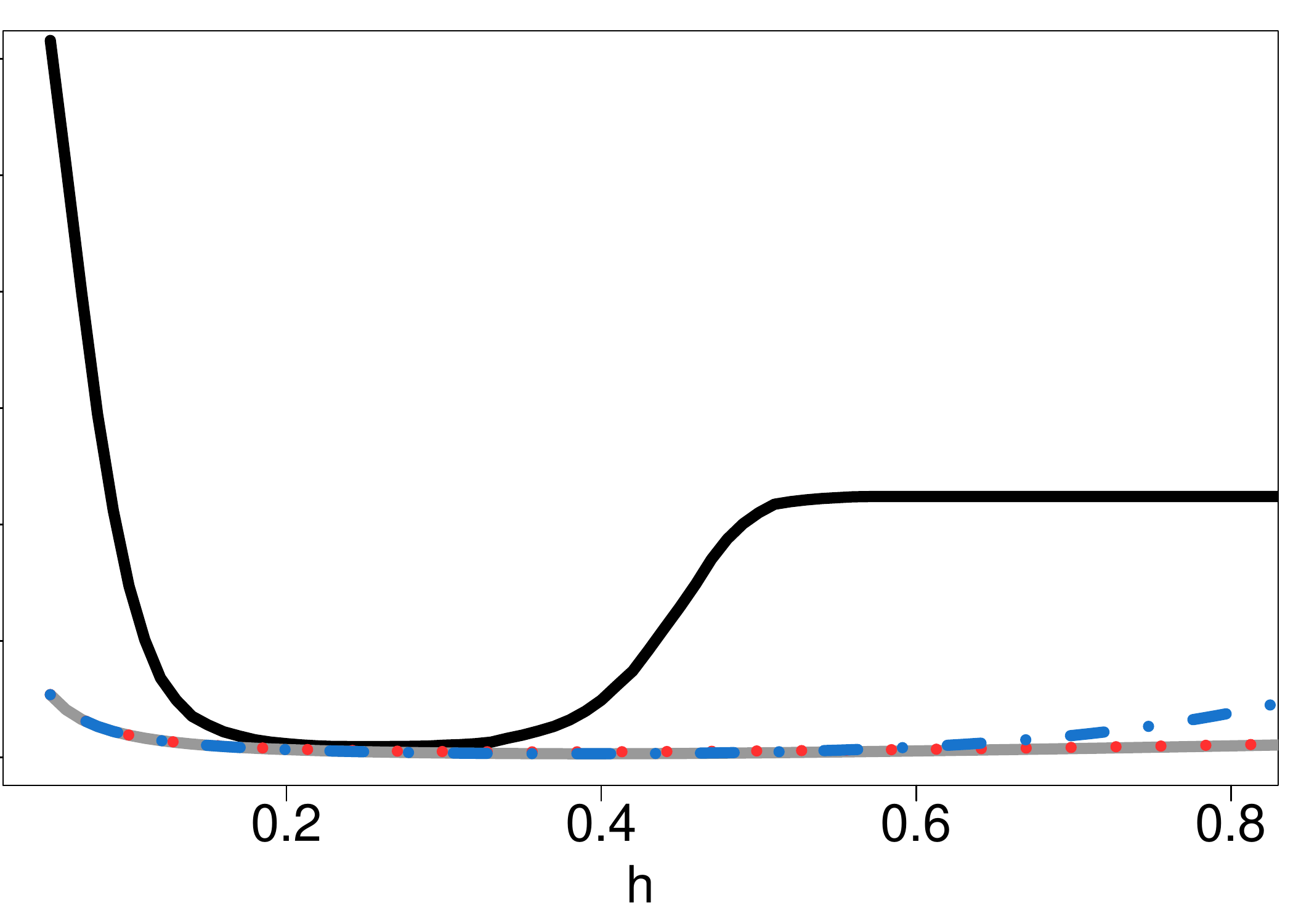}&
   \hspace{-.5cm}  \includegraphics[height=2.5cm]{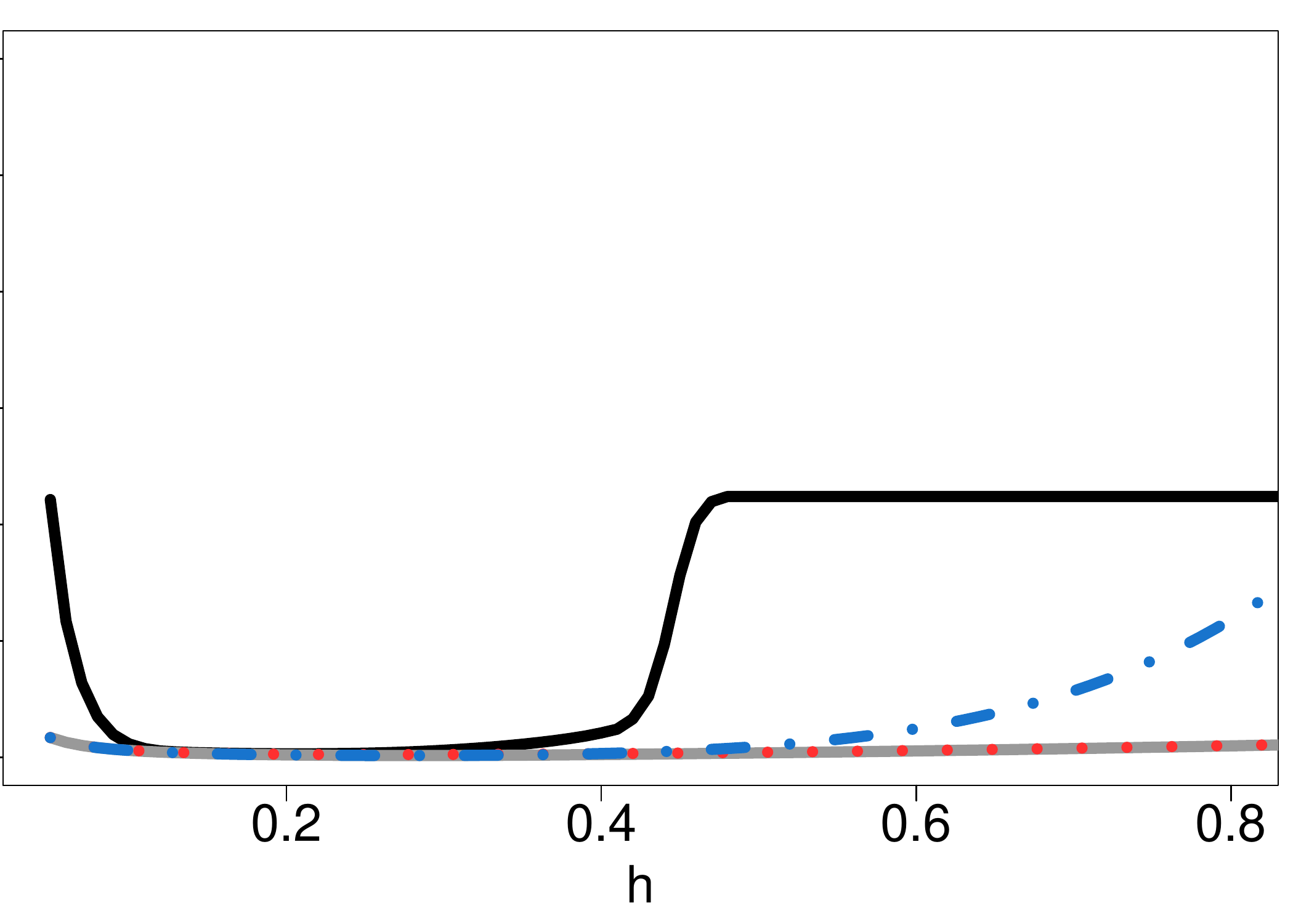}\\
$h_{\rm EDM}$ & 0.039 & 0.009 & 0.003 \\
$h_{\rm AEDM}$ & 0.187 & 0.040 &  0.005\\
  $h_{\rm MISE, 1} $ & 0.040 & 0.015 & 0.005 \\	
    \hline
 $\hat{h}_{\rm AEDM}$ & 0.066 (0.088) & 0.018 (0.054) & 0.011 (0.037) \\
 $\hat{h}_{\rm AB1}$ & 0.057 (0.077) & 0.014 (0.032) & 0.009 (0.033) \\
  $\hat{h}_{\rm AB2}$ & 0.063 (0.082) & 0.015 (0.035) & 0.009 (0.033) \\
   $\hat{h}_{\rm PI, 1}$ & 0.051 (0.069) & 0.011 (0.014) & 0.005 (0.005)\\
 \hline
\end{tabular}
\end{center}
\end{table}

\begin{table}[H]
\caption{Cf. Table \ref{tab:density8}. Results refer to density M5.}\label{tab:trimodal}
\begin{center}
\begin{tabular}{lccc}
\hline
 & 	
n =100&n=1000&n=10000\\
  \hline
& \hspace{-.5cm}  \includegraphics[height=2.5cm]{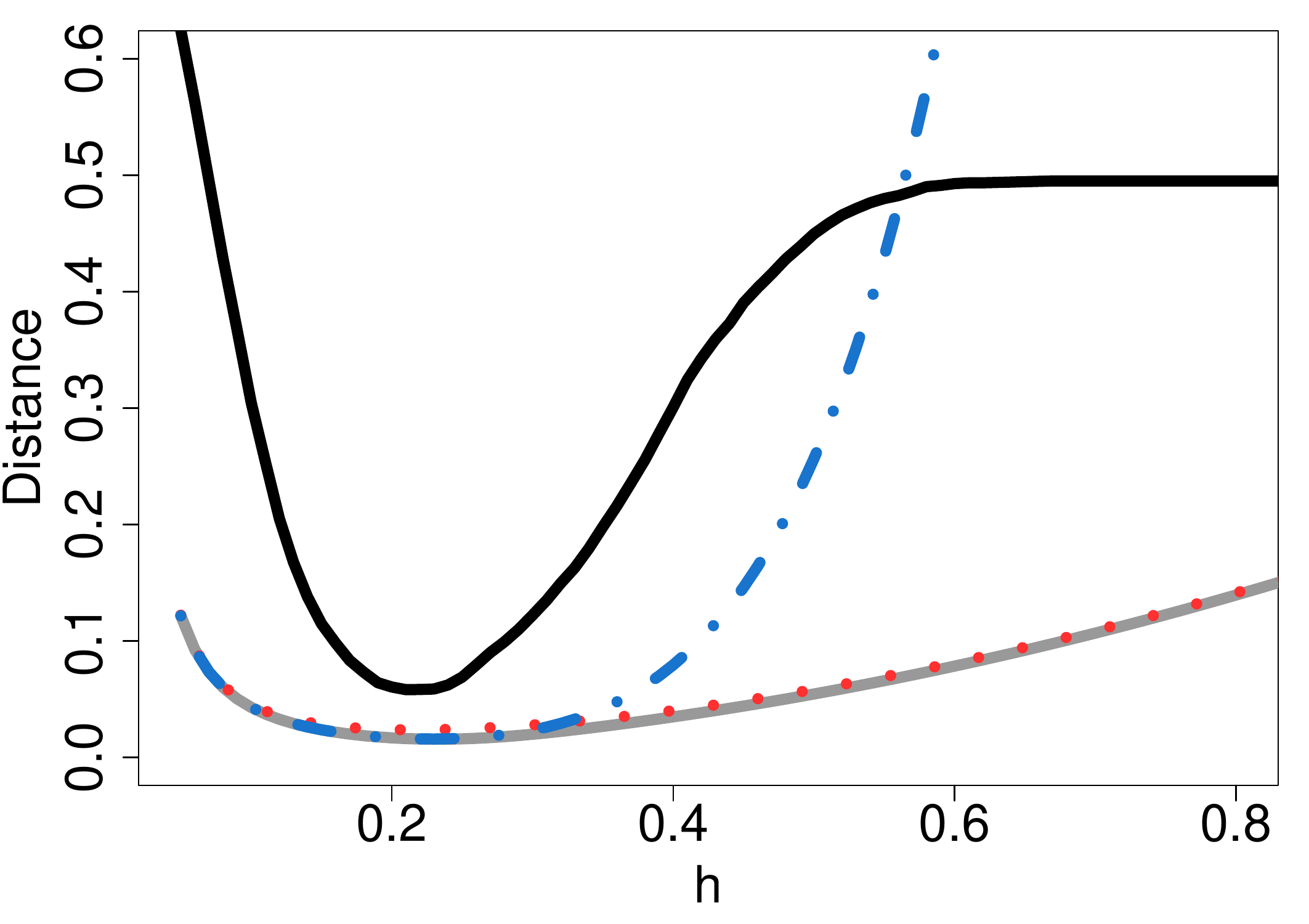}&
   \hspace{-.5cm}  \includegraphics[height=2.5cm]{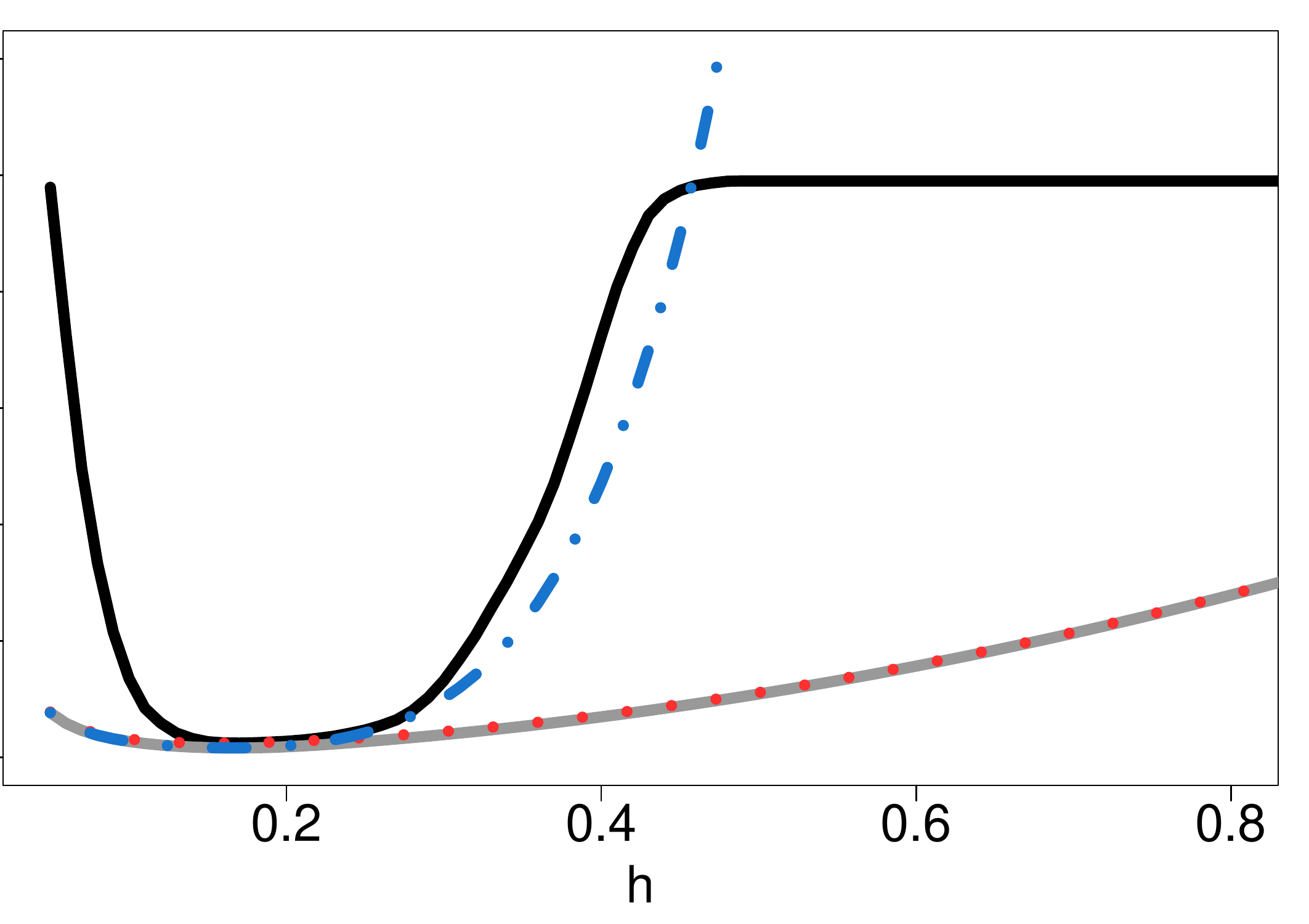}&
   \hspace{-.5cm}  \includegraphics[height=2.5cm]{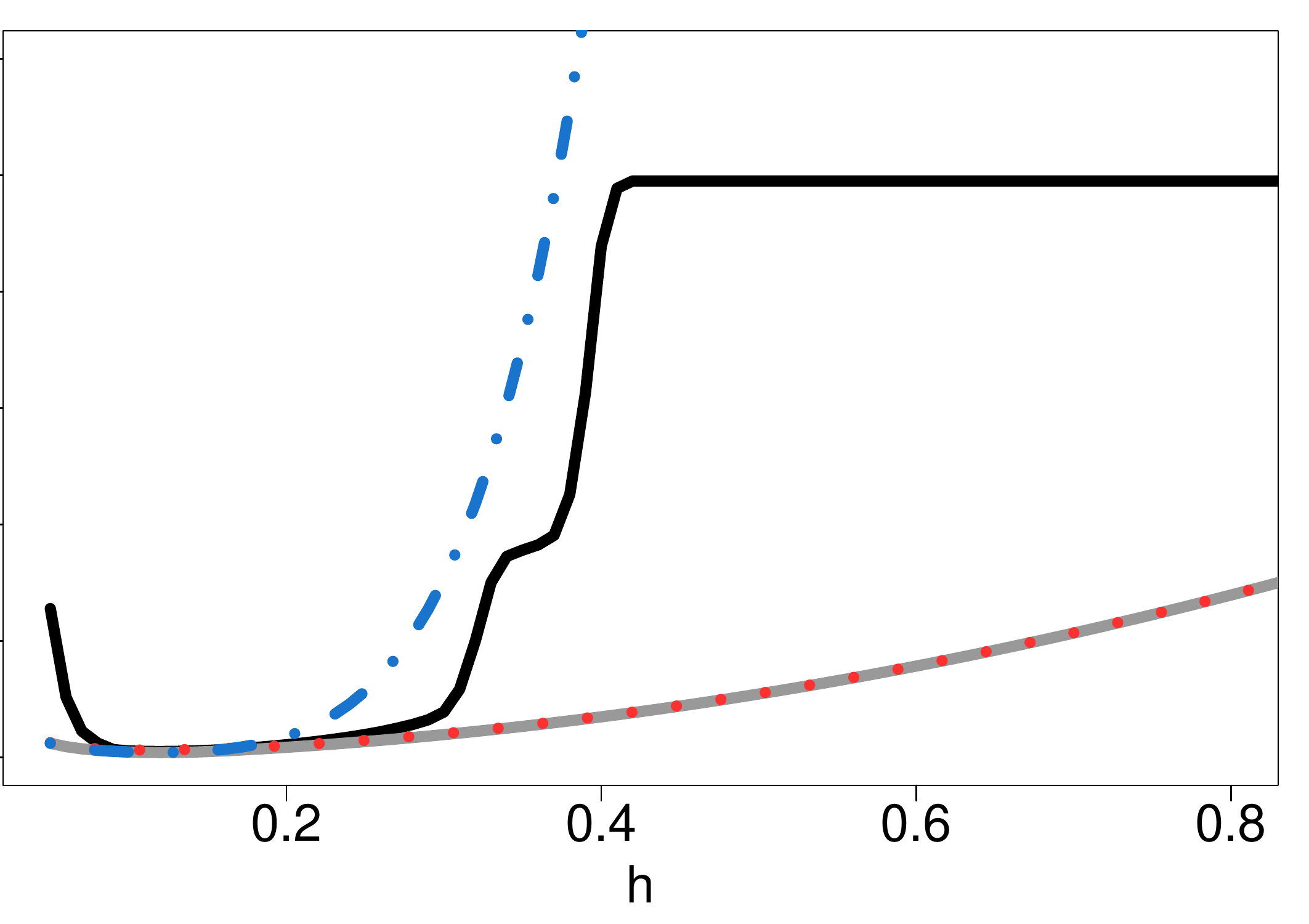}\\
$h_{\rm EDM}$ & 0.058 & 0.012 & 0.005 \\
$h_{\rm AEDM}$ &  0.059 & 0.012 & 0.005 \\
  $h_{\rm MISE, 1} $& 0.058 & 0.013 & 0.005 \\	
    \hline
 $\hat{h}_{\rm AEDM}$ & 0.222 (0.184) & 0.015 (0.018) & 0.005 (0.003) \\
 $\hat{h}_{\rm AB1}$ & 0.195 (0.180) & 0.013 (0.010) & 0.005 (0.003) \\
  $\hat{h}_{\rm AB2}$ & 0.212 (0.184) & 0.014 (0.011) & 0.005 (0.003) \\
   $\hat{h}_{\rm PI, 1}$ & 0.179 (0.158) & 0.013 (0.009) & 0.005 (0.003) \\
 \hline
\end{tabular}
\end{center}
\end{table}

As a second goal, we propose new data-based bandwidth selectors specifically designed for modal clustering purposes. The first step consists in estimating the number of local minima, and their location. This is achieved by numerically finding the roots of a pilot estimate of $f^{(1)}$, constructed as the derivative of the kernel density estimator using the plug-in gradient bandwidth $\hat h_{\rm PI,1}$. Then, similarly, using plug-in bandwidths for $f$, $f^{(2)}$ and $f^{(3)}$, we obtain pilot estimates of these involved unknown functions at the estimated local minima. These quantities are subsequently plugged-in in the formulas of the AEDM, AB1 and AB2, and the minimizers of the resulting estimated criteria are found; in the case of the estimated AEDM by numerical minimization, and according to expressions \eqref{eq:hbound1} and \eqref{eq:hbound2} for AB1 and AB2 respectively. The data-based bandwidths thus obtained are denoted $\hat h_{\rm AEDM}$, $\hat h_{\rm AB1}$ and $\hat h_{\rm AB2}$, respectively.

Occasionally (although rarely) the first step in the procedure above yielded a single mode, and then the AEDM was undefined. In those cases, and according to the rationale exposed in Remark \ref{rem:r1}, a sensible choice for $h$ is the \textit{critical bandwdith} proposed by \cite{Silverman81},
$$\hat h_{\rm crit} = \inf \{h>0 : \hat f_h (\cdot) \mbox{ has exactly one mode}\},$$
so in that case we set $\hat h_{\rm AEDM}=\hat h_{\rm AB1}=\hat h_{\rm AB2}=\hat h_{\rm crit}$. 

Tables \ref{tab:density8} to \ref{tab:trimodal} also contain the Monte Carlo averages and standard deviations of the distances in measure obtained when performing modal clustering using the bandwidth selectors $\hat h_{\rm AEDM}$, $\hat h_{\rm AB1}$ and $\hat h_{\rm AB2}$. For completeness, their performance is also compared to that of $\hat h_{\rm PI,1}$, which so far probably represents their most sensible competitor in the clustering framework (see \cite{chaconmonfort}).

In general, $\hat h_{\rm AB1}$ and $\hat h_{\rm AB2}$ led to more accurate clusterings than $\hat h_{\rm AEDM}$, with a slight preference for $\hat h_{\rm AB1}$. The gradient-based bandwidth $\hat h_{\rm PI,1}$, in turn, not only produces competitive results, but its Monte Carlo average distance in measure appears lower than the one produced by the asymptotic EDM minimizers. In fact, a deeper insight into the standard errors of the obtained distances shows that $\hat h_{\rm AEDM}$, as well as $\hat h_{\rm AB1}$ and $\hat h_{\rm AB2}$, produce more variable results. The higher variability
seems to be due to the sensitivity of the minimizers to the plugged in pilot estimates, which strongly depend on local features of the density.
Some further investigations, not fully reported here, suggest that the main responsible for this behaviour is not the pilot estimate of the local minima but the pilot density derivatives estimates at the minimum points. Also, due to the use of different pilot bandwidths to estimate the unknown $m_j$, $f^{(2)}$, and $f^{(3)}$, it may occur, indeed, that $\hat{f}^{(2)}(\hat m_j)$ assumes even negative values. On the other hand,  while relying as well on some plug-in estimates, the gradient-based bandwidth $\hat h_{\rm PI,1}$ produces more robust clusterings, as the quantities to be estimated refer conversely to global features of the density. As expected, this diverging behavior tends to vanish with increasing sample size since the asymptotic approximations improve. As a confirmation, with $n=10000$, all the considered bandwidths perform comparably.

\section{Multidimensional generalization}\label{sec:multidim}

The concepts discussed so far refer to the one-dimensional setting where a mathematically rigorous treatment is feasible. The multidimensional generalization poses some difficulties since obtaining an asymptotic approximation of the EDM appears far from trivial. Hence, in order to gain some insight into the problem of selecting the amount of smoothing for nonparametric clustering in more than one dimension, some numerical comparisons are performed assuming the true density as known.

Denote by $f: \mathbb R^d \rightarrow \mathbb R$ the true density function and by
 	 		\begin{equation}
 \hat{f}_\mathbf{H}(\mathbf{x})= \frac1n \sum_{i=1}^n |\mathbf{H}|^{-1/2}K\left(\mathbf{H}^{-1/2}(\mathbf{x} - \mathbf{X}_i)\right)\; ,
\end{equation}
its kernel estimate based on a sample $\mathbf{X}_1,\dots,\mathbf{X}_n$ and indexed by a symmetric positive definite $d\times d$ bandwidth matrix $\mathbf{H}$. The problem of bandwidth selection is considered by studying the EDM between the clustering induced by the kernel estimate $\hat{\mathscr C}_\mathbf{H}$ and the ideal population clustering $\mathscr{C}_0$. These clusterings are not so easily identifiable as in the unidimensional setting, due to the arbitrary forms that the cluster boundaries may adopt, however an approximation of the distance in measure $d(\hat{\mathscr C}_\mathbf{H},\mathscr{C}_0)$ can be computed by resorting to a discretization scheme as follows (see \cite{chaconmonfort} for further details):

 	\begin{enumerate}
 		\item Take a grid over the sample space and rule the grid by considering hyper-rectangles centered at each grid point.
 		\item Assign a cluster membership to each grid point by running a population version of the mean-shift algorithm i.e. using the true density. This produces a discretized version of $\mathscr{C}_0$.
 		\item Similarly, obtain the data-based partition $\hat{\mathscr C}_\mathbf{H}$ induced by $\hat{f}_\mathbf{H}.$
 		\item Compute the probability mass of each single hyper-rectangle in  $\mathscr{C}_0$. 
 		\item Compute the distance in measure as in (\ref{eq:distmeas}) where the involved probabilities are evaluated based on the previous step.
 	\end{enumerate}
 	
For the multidimensional simulation study, a total of $B=1000$ samples for each of the sizes $n \in \{100, 1000\}$ were generated from the bivariate densities whose contour plots are shown in Figure \ref{fig:densities2} and described in Appendix \ref{App:settings}. The densities have been chosen to generalize the settings M1 and M5 included in the univariate study.
\begin{figure}[bt]
\begin{center}
\begin{minipage}{0.4\textwidth}
	\includegraphics[scale=.5]{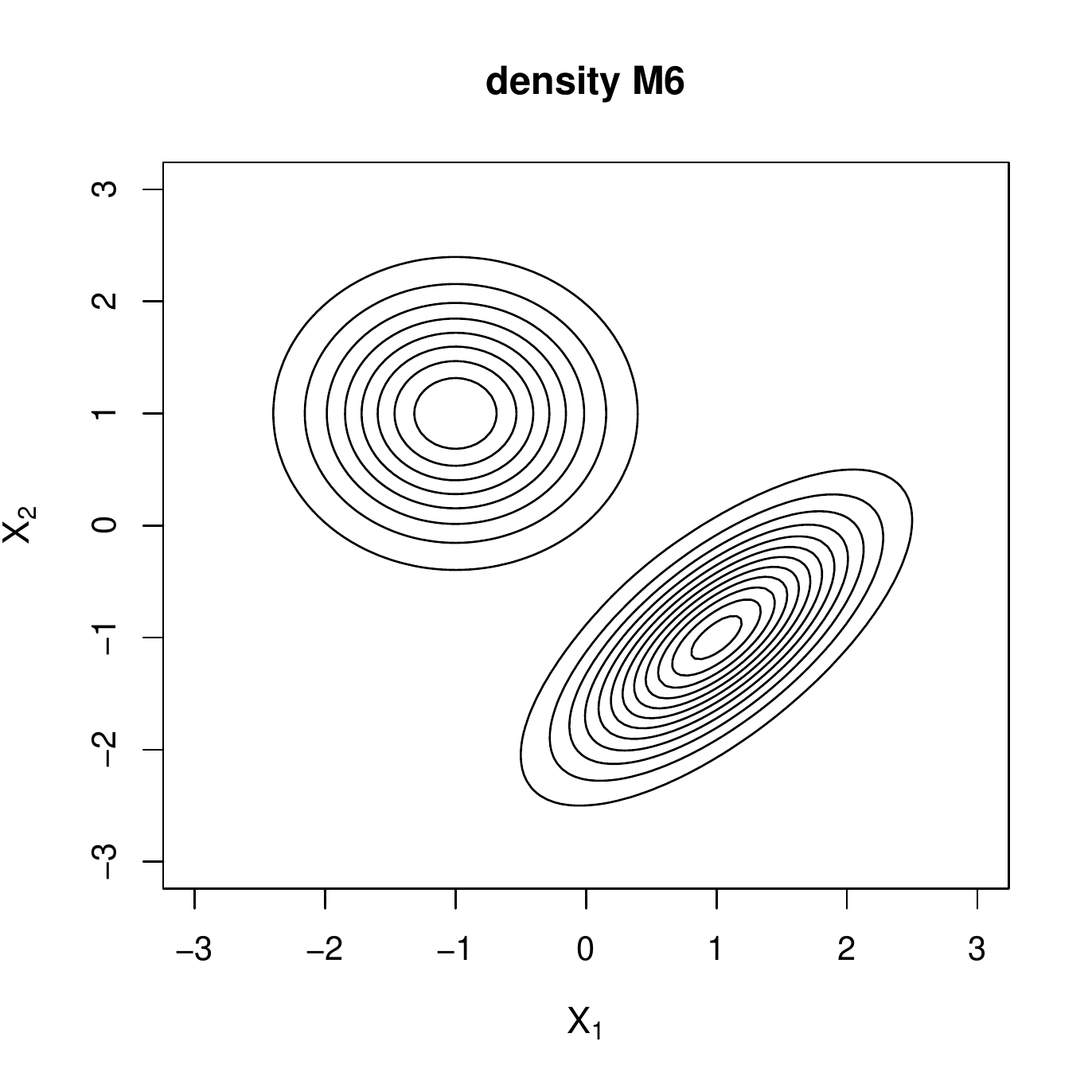}
	\end{minipage}
\hspace{0.1cm}\begin{minipage}{0.4\textwidth}
	\includegraphics[scale=.5]{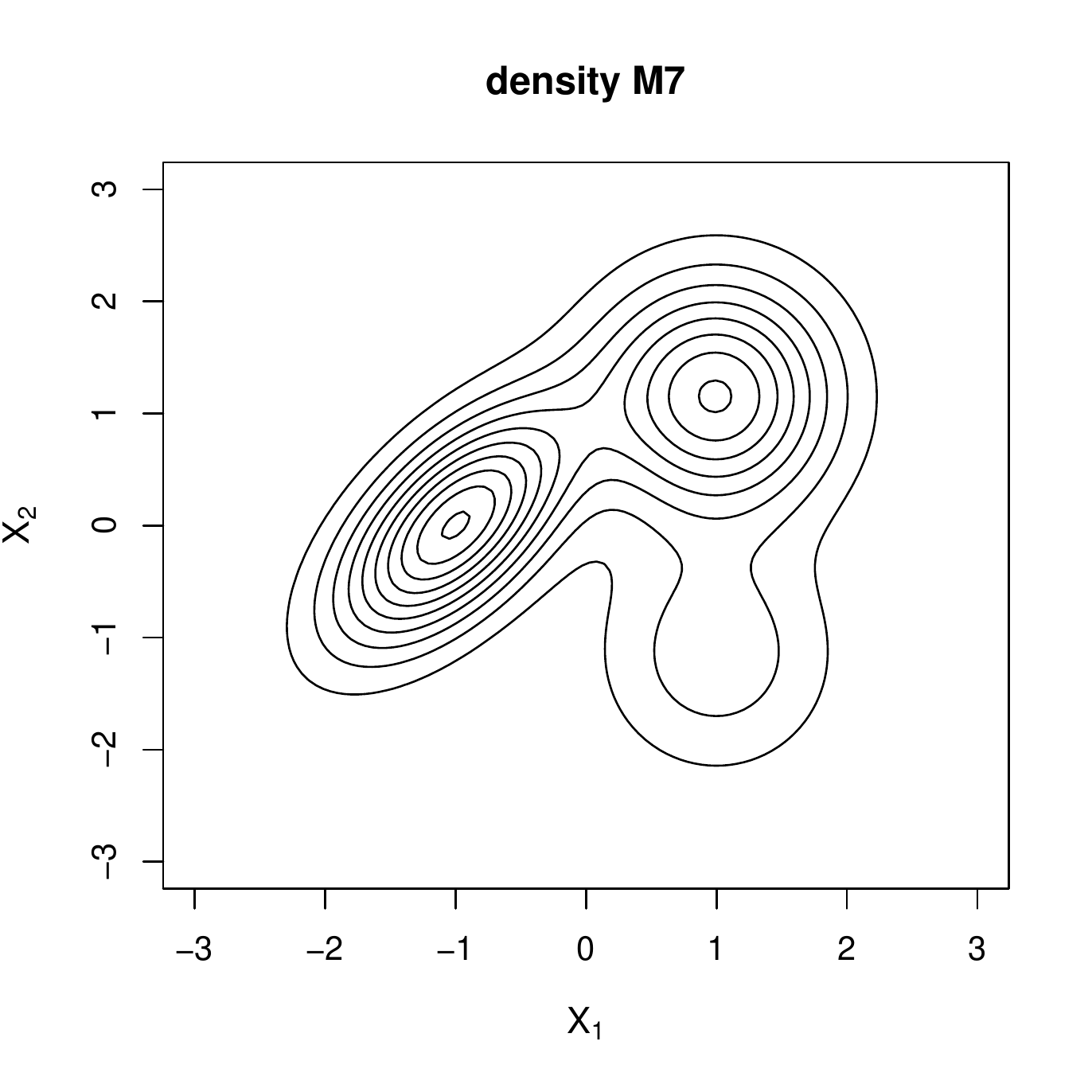}
\end{minipage}
\end{center}
	\caption{Bivariate density functions selected for simulations.}\label{fig:densities2}
\end{figure}

Three different parametrizations for the bandwidth matrix were considered: a scalar bandwidth $\mathbf H=h^2\mathbf I,$ with $\mathbf I$ the identity matrix, a diagonal bandwidth $\mathbf H=\text{diag}(h_1^2,h_2^2)$, and a full, unconstrained bandwidth matrix $\mathbf H.$ For density and density derivative estimation, \cite{wandjones93} and \cite{chacon_etal11} showed that the use of the simplest scalar bandwidth can be quite detrimental in practice, a diagonal bandwidth may suffice in some scenarios, but in general it is advantageous to employ unconstrained bandwidth matrices (see also \cite{chaconduongbook}). However, such results have never been obtained in a modal clustering framework; thus one of the goals of this simulation study is to examine how the bandwidth matrix parametrization affects the performances of the procedures.

Using the synthetic samples from each density in the study, it was possible to obtain a Monte Carlo estimate of the (discretized version of the) EDM, which was then minimized over the class of scalar, diagonal and unconstrained bandwidth matrices. The EDM was computed also for the MISE-optimal bandwidth for density gradient estimation over the same matrix classes. In both cases, the true density as well as all the involved quantities were assumed to be known. The EDM minimizers were determined numerically, by running the procedure over a grid of sensible values of the entries, while the optimal matrices for gradient estimation were determined as in \cite{chacon_etal11}.

The results are reported in Tables \ref{tab:density8_mult} and \ref{tab:density9_mult}. Clustering based on the optimal bandwidth according to the EDM is very accurate in both of the considered examples, and improves considerably for increasing sample size. The use of more complex bandwidth parametrizations does not seem worth for modal clustering since results obtained with a full, unconstrained bandwidth matrix are comparable with those obtained with a scalar bandwidth, while the latter requires a substantially smaller computational effort.

In the multidimensional setting, the gradient bandwidth is quite competitive in terms of EDM, as in the univariate case. Again the comparable performance of unconstrained bandwidth matrices does not seem to justify the use of more complex parametrizations.

\begin{table}[t!]
\caption{Minimum EDM associated with a density estimate with bandwidth matrix $\mathbf H$ selected to minimize the EDM ($\mathbf H_{\rm EDM}$) and the MISE for gradient estimation ($\mathbf H_{\rm MISE,1}$). Different parametrizations for $\mathbf H$ are considered. In both cases, the true density as well as all the involved quantities are assumed to be known.
Results refer to density M6.  }\label{tab:density8_mult}
\begin{center}
\begin{tabular}{lcccc}
&\multicolumn{2}{c}{$\mathbf H_{\rm EDM}$} &\multicolumn{2}{c}{$\mathbf H_{\rm MISE,1}$}\\
\hline
 & n =100&n=1000 & n =100&n=1000\\
  \hline
 $ \left( \begin{array}{cc}
h^2 & 0 \\
0 & h^2
\end{array} \hspace{.25cm}\right) \; \; \; \;$ & 0.006 & 0.004 & 0.064 & 0.040 \\	
 $\left( \begin{array}{cc}
h_1^2 & 0 \\
0 & h_2^2
\end{array} \hspace{.25cm}\right) \; \; \; \;$ &0.006& 0.004 & 0.064 & 0.040  \\
 $\left( \begin{array}{cc}
h_1^2 & h_{12} \\
h_{12} & h_2^2
\end{array} \right) \; \; \; \;$ & 0.005& 0.003& 0.042 & 0.024 \\
\hline
\end{tabular}
\end{center}
\end{table}

\begin{table}[H]
\caption{Cf. Table \ref{tab:density8_mult}. Results refer to density M7.}\label{tab:density9_mult}
\begin{center}
\begin{tabular}{lcccc}
&\multicolumn{2}{c}{$\mathbf H_{\rm EDM}$} &\multicolumn{2}{c}{$\mathbf H_{\rm MISE,1}$}\\
\hline
 & n =100& n=1000 & n =100&n=1000\\
  \hline
 $ \left( \begin{array}{cc}
h^2 & 0 \\
0 & h^2
\end{array} \hspace{.25cm}\right) \; \; \; \;$ & 0.114 & 0.044 & 0.116 & 0.054  \\	
 $\left( \begin{array}{cc}
h_1^2 & 0 \\
0 & h_2^2
\end{array} \hspace{.25cm}\right) \; \; \; \;$ &0.114& 0.042 & 0.115 & 0.055  \\
 $\left( \begin{array}{cc}
h_1^2 & h_{12} \\
h_{12} & h_2^2
\end{array} \right) \; \; \; \;$ & 0.110 & 0.040 & 0.121 & 0.054 \\
\hline
\end{tabular}
\end{center}
\end{table}

\section{Conclusions}\label{sec:conclusion}

The modal clustering methodology provides a framework to perform cluster analysis with a clear and explicit population goal. It allows clusters of arbitrary shape and size, which can be captured by means of a nonparametric density estimator. In this context, the distance in measure represents a natural and easily interpretable error criterion. Therefore, in this paper we have presented an asymptotic study of this criterion for the case where density estimates of kernel type are employed to obtain a whole-space clustering via the mean shift algorithm.

Our asymptotic approximations are useful to gain insight into the fundamental problem of bandwidth selection for modal clustering and, at the same time, serve as the basis to propose practical data-based bandwidth choices specifically designed for clustering purposes.

The finite-sample performance of the new proposals was investigated in a thorough simulation study, and compared to the oracle bandwidths i.e. the optimal choices when the true population is fully known. The gradient bandwidth, designed for the closely related problem of density gradient estimation, was also included as a natural competitor in the study.

The results of this simulation study have suggested that all the methods perform quite satisfactorily, and exhibit a very similar behavior for large sample sizes. For smaller samples, the performance of the gradient bandwidth was rather remarkable, since it obtained comparable or even better results than the new proposals, even without being specifically conceived for modal clustering.

This phenomenon resembles the conclusions obtained in \cite{SGR2014} regarding the related problem of level set estimation. There, it was shown that the traditional bandwidth selectors for density estimation often outperformed more sophisticated methods designed for level set estimation purposes. The common pattern in both situations is that the optimal choices for the specific problems (level set estimation and modal clustering, respectively) depend on very subtle local features of the unknown density function, which are difficult to estimate, so that choices based on a more global, yet somehow related, perspective represent a sensible alternative.


\appendix

\section{Proofs}\label{proofs}

\begin{proof}[Proof of Theorem \ref{thm:AEDM}]
From Theorem 4.1 in \cite{chacon15} it follows that, with probability one, there exists $n_0\in\mathbb N$ such that the kernel density estimator $\widehat f_h$ has the same number of local minima as $f$ for all $n\geq n_0$. Let us denote by $\widehat m_{h,1}<\dots<\widehat m_{h,r-1}$ the local minima of $\widehat f_h$. Then, the expected distance in measure between the data-based clustering $\widehat{\mathscr C}_h$ and the population clustering $\mathscr C_0$ can be written as
\begin{equation}\label{eqEDMF} {\rm EDM}(h)=\sum_{j=1}^{r-1}\mathbb E|F(\widehat m_{h,j})-F(m_j)|.\end{equation}

Write, generically, $\widehat m$ and $m$ for any of the estimated and true local minima. A Taylor expansion with integral remainder allows to write $$F(\widehat m)-F(m)=(\widehat m-m)\int_0^1f\big(m+t(\widehat m-m)\big)dt.$$
The assumptions imply that $\widehat m\to m$ almost surely \citep[see, for instance,][]{Ro88} and, since $f$ is bounded and continuous, this readily yields $\int_0^1f\big(m+t(\widehat m-m)\big)dt\to f(m)$ almost surely, which entails that $\mathbb E|F(\widehat m)-F(m)|\sim f(m)\mathbb E|\widehat m-m|$. The result then follows from Equation (2.6) in \cite{GH95}, where the asymptotic form of $\mathbb E|\widehat m-m|$ is given.
\end{proof}

\begin{proof}[Proof of Lemma \ref{lem:bound2}]
From $\psi(\mu,\sigma^2)=\sigma\psi(\mu/\sigma,1)$, it suffices to show that $\psi(u,1)\leq(2/\pi)^{1/2}+(2\pi)^{-1/2}u^2$ for $u\geq0$. From the definition of $\psi$, this is equivalent to proving that $\alpha(u)\leq1$, where $\alpha(u)=e^{-u^2/2}+u\int_0^ue^{-z^2/2}dz-u^2/2$. Since $\alpha(0)=1$, it is enough to show that $\alpha$ is nonincreasing, but this immediately follows from the fact that $\alpha'(u)=\int_0^ue^{-z^2/2}dz-u$.
\end{proof}

\section{Parameter settings}\label{App:settings}
In the following the parameter settings of the densities selected for the simulations are presented. Since all the densities are mixture of Gaussian models, we adopt the usual notation where, for a given $k$ component, $\pi_k$ represent the \emph{k-th} mixture weight, $\mu_k$ and $\sigma_k^2$ ($\Sigma_k$ for the bivariate models) the mean and variance (covariance matrix).

\subsection{Unidimensional parameter settings}
\subsubsection{Density M1}
\vspace{-0.5cm}

\begin{table}[H]\centering
\hspace{2cm}
\begin{minipage}{0.4\textwidth}
\begin{tabular}{ccccc}
  \hline
 & Components & $\pi_k$ & $\mu_k$ & $\sigma_k^2$ \\
  \hline
& 1 & 0.75 & 0.00 & 0.83\\
 &  2 & 0.25 & 1.37 & 0.09\\
   \hline
\end{tabular}
\end{minipage}
\end{table}

\subsubsection{Density M2}

\begin{figure}[H]\centering
\hspace{2cm}
\begin{minipage}{0.4\textwidth}
\begin{tabular}{ccccc}
  \hline
 & Components & $\pi_k$ & $\mu_k$ & $\sigma_k^2$ \\
  \hline
& 1 & 0.45 & -0.93 & 0.22 \\
 & 2 & 0.45 & 0.93 & 0.22 \\ 
 & 3 & 0.1 & 0.00 & 0.04 \\ 
   \hline
\end{tabular}
\end{minipage}
\end{figure}

\subsubsection{Density M3 }

\begin{figure}[H]\centering
\hspace{2cm}
\begin{minipage}{0.4\textwidth}
\begin{tabular}{ccccc}
  \hline
 & Components & $\pi_k$ & $\mu_k$ & $\sigma_k^2$ \\
  \hline
& 1 & 0.5 & -0.74 & 0.14\\
 & 2 & 0.3 & 0.37 & 0.55\\
 & 3 & 0.2 & 1.47 & 0.14\\
   \hline
\end{tabular}
\end{minipage}
\end{figure}

\subsubsection{Density M4}

\begin{figure}[H]\centering
\hspace{2cm}
\begin{minipage}{0.4\textwidth}
\begin{tabular}{ccccc}
  \hline
 & Components & $\pi_k$ & $\mu_k$ & $\sigma_k^2$ \\
  \hline
& 1 & 0.15 & 0.00 & 0.44\\
 & 2 & 0.15 & -0.33 & 0.19\\
 & 3 & 0.5 & -0.99 & 0.14\\
 & 4 & 0.2 & 1.32 & 0.19\\
   \hline
\end{tabular}
\end{minipage}
\end{figure}

\subsubsection{Density M5}

\begin{figure}[H]\centering
\hspace{2cm}
\begin{minipage}{0.4\textwidth}
\begin{tabular}{ccccc}
  \hline
 & Components & $\pi_k$ & $\mu_k$ & $\sigma_k^2$  \\
  \hline
& 1 & 0.5 & 0.00 & 0.14\\
 & 2 & 0.35 & 1.28 & 0.14\\
 &  3 & 0.15 & 2.56 & 0.11\\
   \hline
\end{tabular}
\end{minipage}
\end{figure}

\subsection{Bidimensional settings}

\subsubsection{Asymmetric bimodal}

\begin{table}[H]\centering
\begin{tabular}{ccccc}
\hline
 & Components & $\pi_k$ & $\mu_k$ & $\Sigma_k$  \\
  \hline \vspace{0.2cm}
  & 1 & 0.5 & $\begin{pmatrix} 1 \\ -1 \end{pmatrix}$ & $\begin{pmatrix} 0.44 & 0.31 \\ 0.31 & 0.44  \end{pmatrix}$ \\
 & 2 & 0.5 & $\begin{pmatrix} -1 \\ 1 \end{pmatrix}$ & $\begin{pmatrix} 0.44 & 0 \\ 0 & 0.44  \end{pmatrix}$ \\
 \hline
\end{tabular}
\end{table}

\subsubsection{Trimodal}

\begin{table}[H]\centering
\begin{tabular}{ccccc}
\hline
 & Components & $\pi_k$ & $\mu_k$ & $\Sigma_k$  \\
  \hline \vspace{0.2cm}
  & 1 & 0.43 & $\begin{pmatrix} -1 \\ 0 \end{pmatrix}$ & $\begin{pmatrix} 0.36 & 0.25 \\ 0.25 & 0.49  \end{pmatrix}$ \\ \vspace{0.2cm}
 & 2 & 0.43 & $\begin{pmatrix} 1 \\ 1.15 \end{pmatrix}$ & $\begin{pmatrix} 0.36 & 0 \\ 0 & 0.49  \end{pmatrix}$ \\
 & 3 & 0.14 & $\begin{pmatrix} 1 \\ -1.15 \end{pmatrix}$ & $\begin{pmatrix} 0.36 & 0 \\ 0 & 0.49  \end{pmatrix}$  \\
 \hline
\end{tabular}
\end{table}

\normalem
\bibliographystyle{plain}
\bibliography{ccm_arxiv}

\begin{thebibliography}{10}

\bibitem{multimode}
J.~Ameijeiras-Alonso, R.M. Crujeiras, and A.~Rodr{\'\i}guez-Casal.
\newblock Multimode: An r package for mode assessment.
\newblock {\em arXiv preprint arXiv:1803.00472}, 2018.

\bibitem{Bal06}
S.~Ben-David, U.~Von~Luxburg, and D.~P{\'a}l.
\newblock A sober look at clustering stability.
\newblock In {\em International Conference on Computational Learning Theory},
  pages 5--19. Springer, 2006.

\bibitem{chacon15}
J.E. Chac{\'o}n.
\newblock A population background for nonparametric density-based clustering.
\newblock {\em Statistical Science}, 30(4):518--532, 2015.

\bibitem{chacon19}
J.E. Chac{\'o}n.
\newblock Mixture model modal clustering.
\newblock {\em To appear in Advances in Data Analysis and Classification},
  2019.

\bibitem{chaconduong}
J.E. Chac{\'o}n and T.~Duong.
\newblock Data-driven density derivative estimation, with applications to
  nonparametric clustering and bump hunting.
\newblock {\em Electronic Journal of Statistics}, 7:499--532, 2013.

\bibitem{chaconduongbook}
J.E. Chac{\'o}n and T.~Duong.
\newblock {\em Multivariate Kernel Smoothing and Its Applications}.
\newblock Chapman and Hall/CRC, 2018.

\bibitem{chacon_etal11}
J.E. Chac{\'o}n, T.~Duong, and M.P. Wand.
\newblock Asymptotics for general multivariate kernel density derivative
  estimators.
\newblock {\em Statistica Sinica}, pages 807--840, 2011.

\bibitem{chaconmonfort}
J.E. Chac{\'o}n and P.~Monfort.
\newblock A comparison of bandwidth selectors for mean shift clustering.
\newblock {\em arXiv preprint arXiv:1310.7855}, 2013.

\bibitem{chen2016}
Y.C. Chen, C.R. Genovese, and L.~Wasserman.
\newblock A comprehensive approach to mode clustering.
\newblock {\em Electronic Journal of Statistics}, 10(1):210--241, 2016.

\bibitem{cgw2017}
Y.C. Chen, C.R. Genovese, and L.~Wasserman.
\newblock Statistical inference using the morse-smale complex.
\newblock {\em Electronic Journal of Statistics}, 11(1):1390--1433, 2017.

\bibitem{chernoff1964}
H.~Chernoff.
\newblock Estimation of the mode.
\newblock {\em Annals of the Institute of Statistical Mathematics},
  16(1):31--41, 1964.

\bibitem{cuevasetal}
A.~Cuevas, M.~Febrero, and R.~Fraiman.
\newblock Cluster analysis: a further approach based on density estimation.
\newblock {\em Computational Statistics \& Data Analysis}, 36(4):441--459,
  2001.

\bibitem{DG85}
L.~Devroye and L.~Gy\"{o}rfi.
\newblock {\em Nonparametric Density Estimation: the $L_1$ View}.
\newblock Wiley, New York, 1985.

\bibitem{DW2018}
C.~R Doss and G.~Weng.
\newblock Bandwidth selection for kernel density estimators of multivariate
  level sets and highest density regions.
\newblock {\em arXiv preprint arXiv:1806.00731}, 2018.

\bibitem{einbeck2011}
J.~Einbeck.
\newblock Bandwidth selection for mean-shift based unsupervised learning
  techniques: a unified approach via self-coverage.
\newblock {\em Journal of pattern recognition research.}, 6(2):175--192, 2011.

\bibitem{Eal11}
B.~S Everitt, S.~Landau, and M.~Leese.
\newblock {\em Cluster Analysis}.
\newblock John Wiley \& Sons, Inc., 2011.

\bibitem{fukunaga75}
K.~Fukunaga and L.~Hostetler.
\newblock The estimation of the gradient of a density function, with
  applications in pattern recognition.
\newblock {\em IEEE Transactions on information theory}, 21(1):32--40, 1975.

\bibitem{GH95}
B.~Grund and P.~Hall.
\newblock On the minimisation of $lp$ error in mode estimation.
\newblock {\em The Annals of Statistics}, 23(6):2264--2284, 1995.

\bibitem{HM91}
P.~Hall and J.S. Marron.
\newblock Lower bounds for bandwidth selection in density estimation.
\newblock {\em Probability Theory and Related Fields}, 90(2):149--173, 1991.

\bibitem{HW88}
P.~Hall and M.P. Wand.
\newblock On the minimization of absolute distance in kernel density
  estimation.
\newblock {\em Statistics \& probability letters}, 6(5):311--314, 1988.

\bibitem{Hal16}
C.~Hennig, M.~Meila, F.~Murtagh, and R.~Rocci.
\newblock {\em Handbook of cluster analysis}.
\newblock CRC Press, 2015.

\bibitem{clue}
K.~Hornik.
\newblock {\em Clue: Cluster ensembles}, 2018.
\newblock R package version 0.3-55.

\bibitem{jones1992}
M.C. Jones.
\newblock Potential for automatic bandwidth choice in variations on kernel
  density estimation.
\newblock {\em Statistics \& probability letters}, 13(5):351--356, 1992.

\bibitem{KR05}
L.~Kaufman and P.J. Rousseeuw.
\newblock {\em Finding groups in data: an introduction to cluster analysis}.
\newblock John Wiley \& Sons, 2005.

\bibitem{LNN61}
F.C. Leone, L.S. Nelson, and R.B. Nottingham.
\newblock The folded normal distribution.
\newblock {\em Technometrics}, 3(4):543--550, 1961.

\bibitem{meanshiftr}
J.~Lisic.
\newblock {\em MeanShiftR: A Computationally Efficient Mean Shift
  Implementation}, 2018.
\newblock R package version 0.52.

\bibitem{matsumoto02}
Y.~Matsumoto.
\newblock {\em An introduction to Morse theory}, volume 208.
\newblock American Mathematical Soc., 2002.

\bibitem{mcnicholas16}
P.D. McNicholas.
\newblock Model-based clustering.
\newblock {\em Journal of Classification}, 33(3):331--373, 2016.

\bibitem{meila2016}
M.~Meila.
\newblock Criteria for comparing clusterings.
\newblock In C.~Hennig, M.~Meila, F.~Murtagh, and R.~Rocci, editors, {\em
  Handbook of Cluster Analysis}, pages 619--635. CRC Press, 2016.

\bibitem{menardi16}
G.~Menardi.
\newblock A review on modal clustering.
\newblock {\em International Statistical Review}, 84(3):413--433, 2016.

\bibitem{qiao2018}
W.~Qiao.
\newblock Asymptotics and optimal bandwidth selection for nonparametric
  estimation of density level sets.
\newblock {\em arXiv preprint arXiv:1707.09697}, 2018.

\bibitem{Rsoftware}
{R Core Team}.
\newblock {\em R: A Language and Environment for Statistical Computing}.
\newblock R Foundation for Statistical Computing, Vienna, Austria, 2018.

\bibitem{Ro88}
J.P. Romano.
\newblock On weak convergence and optimality of kernel density estimates of the
  mode.
\newblock {\em The Annals of Statistics}, pages 629--647, 1988.

\bibitem{SGR2014}
P.~Saavedra-Nieves, W.~Gonz{\'a}lez-Manteiga, and A.~Rodr{\'\i}guez-Casal.
\newblock Level set estimation.
\newblock In {\em Topics in Nonparametric Statistics}, pages 299--307.
  Springer, 2014.

\bibitem{samworthwand}
R.J. Samworth and M.P. Wand.
\newblock Asymptotics and optimal bandwidth selection for highest density
  region estimation.
\newblock {\em The Annals of Statistics}, 38(3):1767--1792, 2010.

\bibitem{scrucca16}
L.~Scrucca.
\newblock Identifying connected components in gaussian finite mixture models
  for clustering.
\newblock {\em Computational Statistics \& Data Analysis}, 93:5--17, 2016.

\bibitem{Silverman81}
B.W. Silverman.
\newblock Using kernel density estimates to investigate multimodality.
\newblock {\em Journal of the Royal Statistical Society. Series B
  (Methodological)}, pages 97--99, 1981.

\bibitem{Silverman86}
B.W. Silverman.
\newblock {\em Density estimation for statistics and data analysis}.
\newblock Chapman \& Hall, 1986.

\bibitem{singh87}
R.S. Singh.
\newblock Mise of kernel estimates of a density and its derivatives.
\newblock {\em Statistics \& probability letters}, 5(2):153--159, 1987.

\bibitem{stuetzle03}
W.~Stuetzle.
\newblock Estimating the cluster tree of a density by analyzing the minimal
  spanning tree of a sample.
\newblock {\em Journal of classification}, 20(1):025--047, 2003.

\bibitem{kspackage}
Duong T.
\newblock {\em ks: Kernel Smoothing}, 2018.
\newblock R package version 1.11.3.

\bibitem{thom49}
R.~Thom.
\newblock Sur une partition en cellules associ{\'e}e {\`a} une fonction sur une
  vari{\'e}t{\'e}.
\newblock {\em Comptes Rendus Hebdomadaires des S\'eances de l'Acad\'emie des
  Sciences}, 228(12):973--975, 1949.

\bibitem{vonLuxburg2010}
U.~Von~Luxburg.
\newblock Clustering stability: an overview.
\newblock {\em Foundations and Trends{\textregistered} in Machine Learning},
  2(3):235--274, 2010.

\bibitem{wandjones93}
M.P. Wand and M.C. Jones.
\newblock Comparison of smoothing parameterizations in bivariate kernel density
  estimation.
\newblock {\em Journal of the American Statistical Association},
  88(422):520--528, 1993.

\bibitem{wandjones}
M.P. Wand and M.C. Jones.
\newblock {\em Kernel smoothing}.
\newblock Chapman and Hall/CRC, 1995.

\end{thebibliography}

\end{document}